%% file: neurips_2026.tex
\documentclass{article}

\usepackage[main, final]{neurips_2026}

\usepackage[utf8]{inputenc} 
\usepackage[T1]{fontenc}    
\usepackage{hyperref}       
\usepackage{url}            
\usepackage{booktabs}       
\usepackage{amsfonts}       
\usepackage{nicefrac}       
\usepackage{microtype}      
\usepackage{xcolor}         

\usepackage{enumitem,mathtools,amssymb,comment}
\usepackage{multirow}
\usepackage{amsthm}
\usepackage{subcaption}
\usepackage{listings}
\usepackage{rotating}
\newcommand{\ConceptTerm}[1]{\ensuremath{\text{\texttt{#1}}}}
\newcommand{\ConPAtt}{ConPAtt}
\DeclareMathOperator*{\argmin}{argmin} 

\theoremstyle{plain}
\newtheorem{theorem}{Theorem}[section]

\newtheorem{lemma}[theorem]{Lemma}

\newtheorem{corollary}[theorem]{Corollary}
\theoremstyle{definition}
\newtheorem{definition}[theorem]{Definition}

\theoremstyle{remark}

\usepackage[capitalize,noabbrev]{cleveref}
\crefalias{lemma}{lemma}
\crefname{lemma}{Lemma}{Lemmata}
\Crefname{lemma}{Lemma}{Lemmata}

\title{Attack logics, not outputs: Concept-based Property Attacks for Shortcut Elimination in Neural Networks }

\author{%
  Raik Dankworth \\
  Institute of Software Engineering and Programming Languages\\
  University of Luebeck\\
  \texttt{r.dankworth@uni-luebeck.de} \\
  \And
  Felix Mächtle\\
  Institute of IT Security \\
  University of Luebeck\\
  \texttt{f.maechtle@uni-luebeck.de} \\
  \And
  Thomas Eisenbarth \\
  Institute of IT Security \\
  University of Luebeck\\
  \texttt{thomas.eisenbarth@uni-luebeck.de} \\
  \And
  Geinsa Schwalbe \\
  Institute of Software Engineering and Programming Languages\\
  University of Luebeck\\
  \texttt{gesina.schwalbe@uni-luebeck.de} \\
}

\begin{document}

\maketitle

\begin{abstract}
    \input{abstract} 
\end{abstract}

\input{main}

{\small
\bibliographystyle{abbrv}
\bibliography{main}
}


\appendix
\input{appendix}


\end{document}

%% file: abstract.tex
Deep neural networks (NNs) for computer vision are vulnerable to adversarial attacks, such that minuscule malicious changes to inputs may induce unintuitive outputs.
One key approach to verify and mitigate such robustness issues is to falsify expected output behavior, which enables local robustness guarantees and adversarial training.
Due to the black-box nature of NNs, current attacks only falsify a class of the \emph{final output},
such as flipping from $\ConceptTerm{stop\_sign}$ to $\neg\ConceptTerm{stop\_sign}$.
In this work we argue that robustness analysis should focus more generally on logical inconsistencies in concept-level reasoning rather than solely on output misclassification. We formalize this perspective through concept-based properties, i.e., constraints over human-interpretable concepts such as
$\ConceptTerm{red}\land\ConceptTerm{octogonal}\implies \ConceptTerm{stop\_sign}$
and study the falsification of such properties as a principled form of adversarial attack.
Furthermore, an implementation of concept-based properties on already trained NNs is proposed using techniques from concept-based explainable artificial intelligence (C-XAI).
For that we formally prove that attacks on concept-based properties are expected to have a reduced search space compared to simple class falsification, whilst arguably be more aligned with intuitive robustness targets.
The practical effectiveness of our approach for logically informed adversarial training is shown in an empirical study across multiple models and datasets, including both vision and a security-relevant code processing task.
Our results demonstrate that concept-level falsification generalizes classical adversarial attacks: Retraining on our novel attacks maintains robustness gains as compared to classical adversarial retraining whilst better eliminating semantic biases.

%% file: main.tex
\section{Introduction}

Neural Networks (NNs) excel in processing subsymbolic inputs like images, and are increasingly being considered for use in safety-critical domains~\cite{rech_artificial_2024}.
This makes it crucial to ensure their robust and intuitive generalization, at least around known training cases.
One tool to evaluate vulnerability to malicious attacks are 
Adversarial Attacks (AAs):
These craft inputs that induce incorrect or unexpected predictions of an NN $f$, using minimal modifications to a correctly handled input $\mathbf{x}$ with $\mathbf{y}=f(\mathbf{x})$
\cite{%
suryanto_dta_2022,li_physical-world_2023,hu_adversarial_2023,zheng_physical_2024,
zhu_learning_2024,wang_boosting_2024,ming_boosting_2024
}.
However, existing attacks solely focus on altering the model's final output, i.e., falsify $\forall \mathbf{x}' \in \text{Nbhd}(\mathbf{x})\colon  f(\mathbf{x}') = y$ for some neighborhood $\text{Nbhd}(\mathbf{x})$ around $\mathbf{x}$ such as an $\epsilon$-ball. 
This disregards whether the prediction still conforms to high-level, interpretable properties.
Common examples of known properties are sufficient conditions, e.g., $\ConceptTerm{red}(\mathbf{x})\land \ConceptTerm{octogonal}(\mathbf{x})\implies \ConceptTerm{stop\_sign}(\mathbf{x})$ in traffic sign recognition from images $x$;
and necessary conditions, like $\neg\ConceptTerm{octogonal}(\mathbf{x})\implies\neg \ConceptTerm{stop\_sign}(\mathbf{x})$.
More general, we call rules involving unary predicates that are not available from the NN outputs \emph{concept-based properties}.
Rich semantic rules are known to be well suited for runtime plausibility monitoring \cite{%
schwalbe2022enabling,%
giunchiglia2022roadr%
} and respective correction of NN outputs 
\cite{%
giunchiglia2022roadr,%
ledaguenel2024improving,
badreddine2022logic%
}. 
In particular, they don't constrain the local \emph{output} to be correct, but the underlying \emph{general logical reasoning} locally around the sample.

One reason why falsification of such informative constraints is not considered for attack generation so far is that they require outputs for all involved predicates---not only the available final output, like $\ConceptTerm{stop\_sign}$.
These additional outputs however, might need a considerable amount of training data or hyperparameter tuning if added right away during the training; or, even worse, not all properties and thus not all required concepts might be known at training time due to specification gaps or later domain transfer \cite{willers2020safety}. 
Meanwhile, NNs automatically learn to encode task-related concepts in their intermediate outputs. For example, when trained for \ConceptTerm{stop\_sign} recognition, the NN may implicitly learn to identify \ConceptTerm{octagons}, \ConceptTerm{red} color, and the \ConceptTerm{stop\_label}.
Post-hoc supervised concept-based explainability methods~\cite{bau_network_2017,fong_net2vec_2018,crabbe_concept_2022,oikarinen_clip-dissect_2022} can recover this information in a very sample-efficient manner with minimal additions to the NN structure. We use this as key to access concepts post-hoc.

Altogether, we propose and analyze a general AA goal ---the \emph{Concept-based Property Attack (\ConPAtt)}---that explicitly targets falsification of symbolic \textit{concept-based properties} over non-symbolic inputs.
As we will show, our formulation offers a more general way to define both targeted and untargeted attacks.
As opposed to classical attacks that purely change the output, our attack on $\neg\ConceptTerm{octogonal}\implies\neg\ConceptTerm{stop\_sign}$ can produce an image still classified as \ConceptTerm{stop\_sign}, in which the \ConceptTerm{octogonal} concept is no longer recognized.
This newly allows to uncover failure cases with semantically inconsistent yet high-confidence predictions that are invisible to standard attacks.
A practical benefit we show empirically is that these semantic-aware adversarial examples do not only (1) achieve competitive robustness gains when used for adversarial training, but (2) are also better in eliminating semantic bias as compared to standard AA formulations.
As we show, standard white-box attack techniques can still easily and efficiently be applied, producing meaningful attacks and a more constrained adversarial space as compared to traditional AAs.

\paragraph{Contributions.}
Our main contributions are:
\begin{itemize}[nosep,leftmargin=1em]
    \item We introduce \ConPAtt, a general XAI-supported adversarial attack objective that targets concept-based properties rather than just NN outputs.
    \item We prove that \ConPAtt{} generalizes both classical targeted and untargeted AA formulations, but operates on a same-sized or smaller adversarial space.
    \item We empirically show that \ConPAtt{} enables semantic perturbations that achieve competitive robustness under adversarial retraining, and improved semantic bias correction in vision and language models.
\end{itemize}


\section{Related Work}

\paragraph{Adversarial Attacks} 
AAs generally search within the vicinity of an input sample $x$ for minimally perturbed variants $\tilde{x}=x+\epsilon$ that have a malicious effect on the NN's output \cite{khamaiseh2022adversarial}.
The perturbations can be arbitrary (digital AAs, considered here) \cite{szegedy_intriguing_2014,wu_skip_2019,su_one_2019,wang_enhancing_2021,zhu_learning_2024,wang_boosting_2024}, or
further constrained to realistic changes (physical AAs)
\cite{li_physical-world_2023,hu_adversarial_2023,%
eykholt_robust_2018,liu_perceptual-sensitive_2019,suryanto_dta_2022,zheng_physical_2024}. However, the minimality makes the changes often invisible or difficult to see for humans.
At the methodological level, black-box approaches only require access to NN inputs and outputs \cite{liu_perceptual-sensitive_2019,suryanto_dta_2022,hu_adversarial_2023,huang_safari_2023,wang_rfla_2023,zheng_physical_2024}. White-box attacks, as considered here, instead exploit NN model internals such as the gradient for a more efficient search \cite{szegedy_intriguing_2014,wu_skip_2019,wang_enhancing_2021,li_physical-world_2023,zhu_learning_2024,wang_boosting_2024}.
Generally, AAs can be seen as a subfield of NN verification that falsifies a continuity property \cite{liu2021algorithms,sun2018concolic}.
Thus, usual search, reachability analysis, and---most prominently---optimization techniques are applicable to find or disprove adversarial examples \cite{khamaiseh2022adversarial}.
Regarding types of specifications beyond continuity properties,
Baninajjar et al.~\cite{baninajjar2025formal} recently formally investigated local deviations between two NNs, and approaches such as Scenic \cite{fremont_scenic_2019} and VerifAI \cite{dreossi_verifai_2019} demonstrate how formal specifications can be used to generate and analyze simulation-based scenarios with symbolic inputs. In contrast, our approach targets AAs on non-symbolic image inputs, which prevents the direct use of such tools but similarly requires formal specifications.

\paragraph{Concept-based Explainability}
Concept-based explainability generally aims to associate human-interpretable concepts with representations in NN latent space \cite{lee2025conceptbased,poeta2023conceptbased,schwalbe2022concept}.
This includes understanding which concepts are relevant to the decision and to what extent \cite{wan_nbdt_2020,oikarinen_clip-dissect_2022}, and how these can be accurately recognized in NNs \cite{fong_net2vec_2018,maglogiannis_verification_2021}.
If concept definitions in form of labeled samples are available at training time, ante-hoc approaches~\cite{koh_concept_2020,wan_nbdt_2020,yuksekgonul_post-hoc_2022,oikarinen_label-free_2022,yang_language_2023} can train individual neurons to activate for the concept.
We here instead consider post-hoc approaches: These train a simple model to predict the concept of interest from an NN layer's activation \cite{kim2018interpretability}. Other than single-neuron-associations \cite{bau2017network,olah2017feature}, or complex models \cite{crabbe2022concept,zhang2021invertible}, linear models considered here \cite{fong2018net2vec,kim2018interpretability,graziani2018regression,mikriukov2025local} pose a good tradeoff between capturing the entanglement of representations \cite{fong2018net2vec,dreyer2024pure}, interpretability \cite{kim2018interpretability}, and favorably simple representation of the concept as halfspace in the NN's latent space.

\paragraph{XAI and Verification}
Prior work has shown that concept-based explanation methods are vulnerable to adversarial attacks.
Perturbations can mislead attribution \cite{slack_fooling_2020} and concept-based tools \cite{mikriukov_unveiling_2024,brown_making_2023}, and adversarial examples significantly alter the internal concept composition of NNs \cite{mikriukov_unveiling_2024}, confirming the general fragility of interpretability methods \cite{ghorbani_interpretation_2019}.
However, these studies target concepts in isolation, without considering their joint relation to model predictions.

Beyond highlighting vulnerabilities, concept outputs have also been used for verification.
Mangal et al. \cite{mangal_concept-based_2024} employed vision–language models to check concept-based properties.
While expressive, this approach relies on semantic similarity in multimodal embeddings (e.g., CLIP \cite{radford_learning_2021}), which can introduce linguistic ambiguity as well as imprecision for similar terms with small visual differences, e.g., \ConceptTerm{circle} versus \ConceptTerm{octagon}.
Moreover, it is restricted to the latent space of a specific layer, although simple visual concepts may predominantly appear earlier and diminish in later layers.
Cheng et al.~\cite{cheng_towards_2020} proposed specifications close to the output layer, but without decomposing them into underlying concepts and by employing an additional NN.
Semantic losses \cite{xu2018semantic,badreddine2022logic} like logic tensor networks \cite{badreddine2022logic} suggest to directly train concept-based rules into the network. These techniques, however, are only used for updating the NN, not for verification as done in \cite{schwalbe2022enabling}, and not for AAs. Furthermore, they rely on concepts being direct outputs of the NN. Even further decoupling the verification from the NN's learned representations and thus exacerbating training efforts, Xie et al.~\cite{xie2022neurosymbolic} even trained completely separate NNs for predicting the concepts. 
Our work also directly addresses the relationship between concepts and model outputs, a perspective that has received little attention so far \cite{schwalbe2022enabling}. However, similar to the verification testing techniques from \cite{cheng_towards_2020,schwalbe2022enabling}, we suggest to keep training and verification efforts low by using faithful explainability techniques to access concept predictions, and we newly apply the setup to AAs.

\section{Background}

\paragraph{Adversarial Attacks}
Let $\mathbf{x} \in \mathcal{X}$ be a real image, $\mathbf{y} \in \mathcal{Y}$ be its true label, and $f\colon \mathcal{X} \to \mathcal{Y}$ be a NN.
An AA seeks an adversarial example $\mathbf{x}^{\text{adv}} \coloneqq \mathbf{x} + \epsilon \in \mathcal{X}$ so that its output is (sufficiently) different from the original, 
and the perturbation~$\epsilon$ is minimal to an objective function $o$ (usually the L1, L2, or L-infinity norm on the input for digital attacks).
Sufficient difference can be formulated in terms of a $\mathbf{y}$-specific partition of the output set $\mathcal{Y}$ into a benign output set
$\mathcal{Y}^+
\subset \mathcal{Y}$ with $f(\mathbf{x})\in\mathcal{Y}^+$, and a malicious one $\mathcal{Y}^- \coloneqq \mathcal{Y} \setminus \mathcal{Y}^+$. 
The search for the minimum perturbation~$\epsilon$ then is the optimization problem
\begin{align}
\SwapAboveDisplaySkip
    &{\textstyle\argmin_\epsilon} \> o(\epsilon)  \qquad
    \text{s.t. }  f(\mathbf{x} + \epsilon) \in \mathcal{Y}^-
    \;.
\end{align}
Adversarial attack strategies for classification are commonly categorized as \emph{targeted} or \emph{untargeted}, depending on the definition of the set~$\mathcal{Y}^{-}$:
\[
\mathcal{Y}^{-} =
\begin{cases}
  \{\, \mathbf{y} \in \mathcal{Y} \mid \mathbf{y} \nvDash l \,\}, & \text{(untargeted)}, \\
  \{\, \mathbf{y} \in \mathcal{Y} \mid \mathbf{y} \vDash l' \,\}, & \text{(targeted)},
\end{cases}
\]
where $l$ denotes the true class and $l' \neq l$ the targeted incorrect class.
The satisfaction relation $\mathbf{y} \vDash l$ indicates that instance~$\mathbf{y}$ is assigned label~$l$ by the classifier, while $\mathbf{y} \nvDash l$ denotes the negation thereof.

\paragraph{Post-hoc Concept Extraction} 
Let $C$ be a set of concepts (e.g., $C=\{\ConceptTerm{red}, \ConceptTerm{orthogonal}\}$), and assume a possibly small classification dataset $\mathcal{D}_c = \left((\mathbf{x}_k, \mathbf{y}_{c,k})\right)_k$ is available per concept $c\in C$. Further denote by $f_{i\to j}\colon \mathcal{X}_i\to\mathcal{X}_j$ the NN part that maps from the $i$th to the $j$th layer.
Through linear post-hoc concept extraction, additional concept outputs are added to the NN by attaching for each $c$ a linear classification model $f_{i\to c}\colon \mathcal{X}_i\to \mathcal{C}_c=[0,1]$ to the $i$th hidden layer as illustrated in \Cref{fig:post-hoc-network}. Keeping the NN's weights fixed, the weights of $f_{i\to c}$ are trained on pairs $\left((f_{\to i}(\mathbf{x}_k), \mathbf{y}_{c,k})\right)_k$,
such that $c$'s concept function $f_c = f_{i\to c}\circ f_{\to i}\colon \mathcal{X}\to \mathcal{C}_c$ correctly predicts presence of the concept in an input image.
Note that $f_{i\to c}$ being linear, conveniently makes any subspace $\{v\in \mathcal{X}_i\mid f_{i\to c}(v) > \theta\}$ an affine linear half-space.
In the following, we denote by $f_C = (f_c)_{c\in C}\colon \mathcal{X} \to \mathcal{C} = (\mathcal{C}_c)_{c\in C} $ the complete prediction of all concepts, and by $\mathcal{Z} = \mathcal{Y} \times \mathcal{C}$ the complete output set after attaching the concept outputs.


\begin{figure}
    \centering
    \includegraphics[width=.45\linewidth]{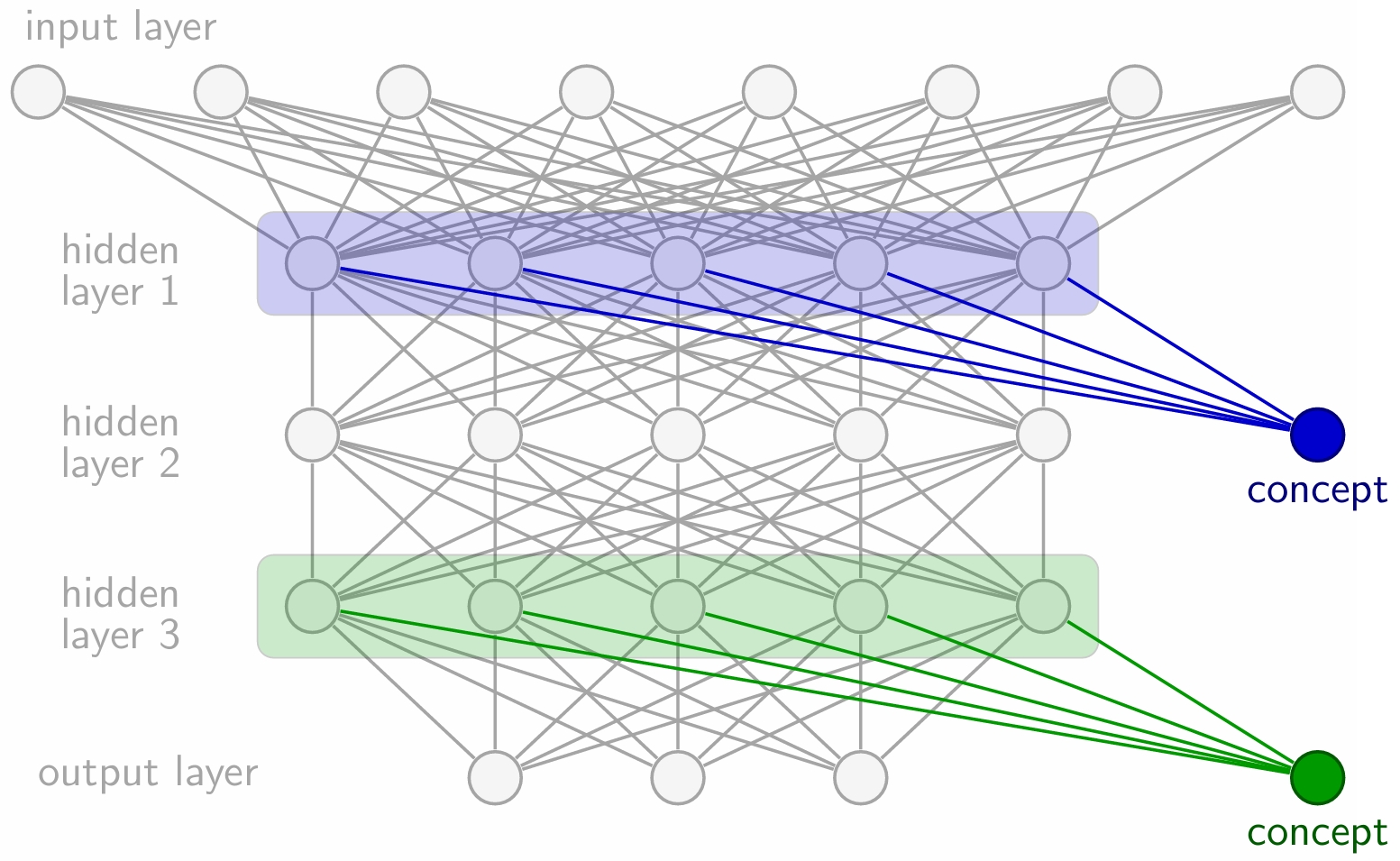}%
    \caption{\textbf{Post-hoc Concept Extraction:}
    Two neurons for concept output (\textcolor{blue}{blue} and \textcolor{green}{green} neurons) are post-hoc added to the trained NN (\textcolor{gray}{gray}) using newly trained connections to hidden layers \textcolor{blue}{1}/\textcolor{green}{3}.
    }
    \label{fig:post-hoc-network}
\end{figure}

\paragraph{t-Norm Fuzzy Logic}
\label{subsec:fuzzy-logic}
The standard Boolean logical connectives (\emph{and} $\wedge$, \emph{or} $\vee$, \emph{not} $\neg$) can only operate on binary truth values in $\mathbb{B}=\{0,1\}$. t-norm fuzzy logics extend the connectives to many-valued truth values in $\mathbb{B}=[0,1]$ using a so-called t-norm $\wedge_t\colon [0,1]\times [0,1]\to [0,1]$ to replace $\wedge$. A valid t-norm must be monotonic, commutative and associative, have a neutral element (the $1$), and match $\wedge$ on Boolean values (0 and 1).
Typical choices for $a\wedge_t b$ are Product ($a\cdot b$), Łukasiewicz ($\max(0,a+b-1)$), and Gödel ($\min(a,b)$) t-norms \cite{hajek_metamathematics_1998}, since these form a generating system for all continuous t-norms.
Given a $\wedge_t$, then the connectives $\neg a \coloneqq 1-a$, $\vee_t,\implies_t\colon[0,1]^2\to[0,1]$ can be derived and maintain desirable properties, together constituting the resulting t-norm logic.

Desirable properties for use of t-norm logic with NN classification outputs are:
(1) The NN typically produces a confidence prediction in $[0,1]$ instead of a Boolean value, which can be propagated by t-norm fuzzy logic to the confidence of entire logical expressions.
(2) The classical piece-wise continuous t-norm logic connectives are also piece-wise differentiable like ReLU activations of NNs. So, they can directly be used in backpropagation and gradient-based adversarial attacks \cite{badreddine2022logic}.

\section{Approach}
In this chapter we first define our new notion of concept-based AAs. Then we show that standard AAs are a special case, and existing attack techniques can easily be adopted to our new attack.

\subsection{Concept-based Property Attacks}

These classical AA types can also be interpreted as special cases of \emph{property attacks}, where class predictions are represented as logical literals.
Using fuzzy logic (see \Cref{subsec:fuzzy-logic}), logical properties over model outputs are evaluated using the satisfaction relation~$\mathbf{y} \vDash \rho$, which denotes that instance~$\mathbf{y}$ satisfies property~$\rho$.
A property attack aims to falsify a given property, that is, $\mathcal{Y}^{-} = \mathcal{Y}^{-}_{\rho} = \{\, \mathbf{y} \in \mathcal{Y} \mid \mathbf{y} \nvDash \rho \,\}.$
Untargeted and targeted attacks correspond to properties $\rho = l$ and $\rho = \lnot l$ for true class $l$, respectively.
This perspective allows adversarial examples to be crafted with higher-order conditions — e.g., enforcing predictions of classes $l_1$ and $l_2$ simultaneously.
The corresponding attacked property is its logical negation: $\rho = \lnot l_1 \lor \lnot l_2$.

The point of view of property attacks can also be applied to NNs that are enriched using XAI techniques.
Additional concept outputs can be used just like the original task output of the NN to define property attacks--- \emph{Concept-based Property Attacks} (\ConPAtt{}).
For denoting the properties we use the following implication-based normal form, which serves as a convenient and intuitive notation throughout this paper.
By Lemma \ref{lemma:implication-form}, every logical expression over concept and task literals can be equivalently rewritten into this form.
For simplicity, we shorten $f_c(\mathbf{x})$ to $c$.

\begin{lemma}
    \label{lemma:implication-form}
    Each  logical expression~$\varphi$ with two disjoint literal sets $C$ and $L$ can be reformulated into a term of conjunctively linked implication terms where antecedents consist only of conjunctively linked literals of $C$, and consequences consist only of disjunctively linked literals of $L$.
\end{lemma}

\begin{proof}
Each logical expression can be reformulated into the conjunctive normal form
$\textstyle
    \varphi \equiv
    \bigwedge_{i }
    (
    \bigvee_{c \in C_i \subseteq C} c 
    \bigvee_{l \in L_i \subseteq L} l
    )
$.
%
Let us introduce two additional variable families $\alpha_i, \beta_i$ that condense the disjunctive subformulas:
\begin{align}
    \alpha_i &\coloneqq \lnot \bigvee\limits_{\mathclap{c \in C_i \subseteq C}} \lnot c  
    \equiv  \bigwedge\limits_{\mathclap{c \in C_i \subseteq C}} c &
    \beta_i &\coloneqq \bigvee\limits_{\mathclap{l \in L_i \subseteq L}} l
\end{align}
The subformulas can be replaced by these variables and the whole logical expression~$\varphi$ reformulates to
$
    \varphi 
    \equiv \bigwedge\limits_{i } \left(\lnot \alpha_{i} \lor \beta_{i}\right)
    \equiv \bigwedge\limits_{i } \left(\alpha_{i} \implies \beta_{i}\right)
$.
\end{proof}

\begin{definition}[Concept-based property]
\label{def:concept-based-property}
A concept-based property~$\varphi$ is a logical expression over two disjoint literal sets $C$---the concept literals---and $L$---the task literals---, admitting the following implication-based representation, where antecedents consist only of conjunctively linked concept literals and consequences consist only of disjunctively linked task literals.
\begin{gather}
    \varphi \coloneqq \bigwedge\limits_{i } \left(\alpha_{i} \implies \beta_{i}\right)\;,
    \quad\text{ with }\quad
    \alpha_i \coloneqq \bigwedge\limits_{\mathclap{c \in C_i \subseteq C}} c
    \quad\text{ and }\quad
    \beta_i \coloneqq  \bigvee\limits_{\mathclap{l \in L_i \subseteq L}} l
\end{gather}
\end{definition}
This implication-based representation is purely notational and neither restricts expressiveness nor alters the results (according to Lemma \ref{lemma:implication-form}).

\begin{definition}[Concept-based Property Attack]
  Let $\varphi$ be a concept-based property and let $o$ be a minimality measure for perturbations $\epsilon$.
  A Concept-based Property Attack (\ConPAtt{}) of $\varphi$ is the search for an $o$-minimal perturbation $\epsilon$ to an input $\mathbf{x}$ transforms $\mathbf{x}$ into an adversarial example~$\mathbf{x}^{\text{adv}}=\mathbf{x}+\epsilon$ which falsifies $\varphi$, i.e.,
  lies in the malicious output set
  \begin{align}
    \mathcal{Z}^- = \mathcal{Z}^-_\varphi =  \lbrace \mathbf{z} \in \mathcal{Z} \mid  \mathbf{z} \nvDash \varphi \rbrace
  \end{align}
\end{definition}

Intuitively, a \ConPAtt{} adversarial example $\mathbf{x}^{\text{adv}}$ to $\phi=(\bigwedge_{c}c\implies\bigvee_{l}l)$ like $\ConceptTerm{red}\wedge \ConceptTerm{octogonal}\implies \ConceptTerm{stop\_sign}$, causes the NN to predict all $c$ as true, and all $l$ as false.
This can happen if (1) some $c$ is predicted true even though it should be false (e.g., \ConceptTerm{red} predicted true even though the change $\epsilon$ turned the sign \ConceptTerm{gray}), and/or (2) some $l$ is predicted negative even though it should be positive (e.g., \ConceptTerm{stop\_sign} flipped to false).

\subsection{\ConPAtt{}s as Generalized Adversarial Attacks}
\label{subsec:generalized-aa}
Note that falsifying one implication term is enough to falsify a concept-based property and thus, it is sufficient to consider one implication $\phi=\alpha\implies \beta$ for an attack.
The set of adversarial example task outputs can be derived from this definition, i.e. $\mathcal{Y}^-_{\varphi} \coloneqq \lbrace \mathbf{y} \in \mathcal{Y}  \mid (\mathbf{y}, \mathbf{c})  \in \mathcal{Z}^-_{\varphi} \rbrace$. Furthermore:

\begin{corollary}\label{thm:special-cases}
  Standard targeted and untargeted AAs are special cases of \ConPAtt{}.
\end{corollary}
\begin{proof}
    First note the two special cases of \ConPAtt{} where only a single task literal is used:
    \begin{enumerate}[nosep, noitemsep]
        \item \textbf{Generalized untargeted AAs}: $\alpha \implies l$.
        \item \textbf{Generalized targeted AAs}: $\alpha \implies \lnot l$.
    \end{enumerate}
    Un-/targeted AAs are generalized un-/targeted AAs with $\alpha \equiv \text{true}$, i.e., no concept restriction.
\end{proof}

A neat property of \ConPAtt{}s is that the search space is generally reduced compared to vanilla AAs:
\begin{theorem}\label{thm:reduced-search-space}
  The task output spaces of adversarial examples for generalized untargeted/targeted AAs are smaller than or equal to those for standard untargeted/targeted AAs.
  \begin{align*}
    \mathcal{Y}^-_{\alpha \implies l} \subseteq \mathcal{Y}^-_{l} \qquad \mathcal{Y}^-_{\alpha \implies \lnot l} \subseteq \mathcal{Y}^-_{\lnot l}
  \end{align*}
\end{theorem}
\begin{proof}
Let us first look at generalized untargeted AA properties like $\alpha \implies l$.
Each adversarial example must lack class prediction~$l$ but requires concept predictions~$\alpha$, i.e., they satisfy the property $\alpha \land \lnot l$.
In contrast to that, standard untargeted AAs only require the misclassification of $l$, i.e. each adversarial example satisfies $\lnot l$ and they accept adversarial examples that do not additionally fulfill $\alpha$.
It follows that the valid output space of adversarial examples for generalized untargeted AAs~$\mathcal{Z}^-_{\alpha \implies l}$ is smaller than or equal to that for standard untargeted AAs~$\mathcal{Z}^-_{l}$ as well as for their valid task output spaces $\mathcal{Y}^-_{\alpha \implies l} \subseteq \mathcal{Y}^-_{l}$.

In this explanation, it does not matter whether both adversarial examples expect a misclassification~$\lnot l$ or a specific task output~$l$.
That is why this relation also applies between generalized targeted AAs and standard targeted AAs, i.e. $\mathcal{Y}^-_{\alpha \implies \lnot l} \subseteq \mathcal{Y}^-_{\lnot l}$.
\end{proof}

\paragraph{\ConPAtt{} Procedure}
\ConPAtt{} can be easily performed with any existing AA approach.
The trick is to use the result of the (partially) differentiable fuzzy operation $\varphi\circ (f, f_C)\colon\mathcal{X}\to \mathbb{B}$ instead of the output of the NN.
This makes \ConPAtt{} a targeted AA with the expected result false or 0 for the adversarial examples.

\section{Evaluation}
\label{sec:evaluation}
In this section, we evaluate the effectiveness of \ConPAtt{} in terms of generalization and robustness under adversarial retraining.
Overall, we aim to assess whether \ConPAtt{} improves robustness beyond standard attacks such as PGD and whether it induces qualitatively different adversarial behavior.
Following our theoretical analysis, we hence explore the intuitive hypotheses that
(1) \ConPAtt{} also practically generalizes AAs in robustification, i.e., adversarial retraining under \ConPAtt{} achieves competitive robustness; and
(2) \ConPAtt{} is superior to the standard AA formulation in correcting \emph{semantic} model bias.
We focus on F1 score as the primary evaluation metric and analyze performance along two axes: (i) generalization on original data after retraining, and (ii) robustness to PGD adversarial examples.
Finally, we complement vision experiments (\cref{subsec:experimental_setup,subsec:results}) with an additional study on a large language model to assess whether the observed effects extend beyond the vision domain (\cref{subsec:llm-evaluation}).

\subsection{Experimental Setup}
\label{subsec:experimental_setup}

\paragraph{Datasets}
We validate and evaluate our approach on two vision benchmarks, MNIST~\cite{deng_mnist_2012}, GTSRB~\cite{Stallkamp2012}, which differ in semantic complexity and can be described by rules.
In addition to these datasets, we construct a biased variant of GTSRB by introducing simple class-specific features (small static coloured rectangle pasted into the image) that create strong spurious correlations, encouraging models to rely on shortcut solutions.  
Both use their standard train--test splits.  
All models are trained with early stopping based on validation loss, and the best checkpoint is used for evaluation.  

\paragraph{Models and Training}
We consider convolutional, residual, and transformer-based architectures to cover various inductive biases, namely LeNet-5~\cite{lecun_gradient-based_1998}, M3~\cite{an_ensemble_2020}, ResNet-8~\cite{he_deep_2016}, VGG-16~\cite{simonyan_very_2015} (ImageNet-pretrained), and a Vision Transformer (ViT)~\cite{dosovitskiy_image_2021}.
Each architecture is trained with multiple seeds.
We follow standard training procedures (details in \Cref{app:cv-setup} of the Appendix).

\paragraph{Concept Mining}
Concepts are learned by training linear classifiers on intermediate feature representations of each model.
These classifiers predict the presence of predefined, task-specific concepts, enabling a structured representation of model internals.
For MNIST, concepts correspond to geometric primitives with an explicit hierarchical structure, while for GTSRB, concepts are grouped into shape, colour, and symbol categories, using a synthetically generated dataset (details in \cref{app:cv-setup}).

\paragraph{Adversarial Attacks}
We focus on PGD~\cite{madry_towards_2018} as a strong classical baseline to generate adversarial examples; extending to additional attacks is left for future work.
An $\epsilon$-sweep on MNIST identifies $\epsilon = 0.3$ as a reliable trade-off between attack success and perceptual distortion, which we use across all experiments.

\ConPAtt{} operates on a predefined set of logical rules using G\"odel t-norm semantics.
We consider two complementary strategies: our approach, decreasing rule compliance (R$\downarrow$), and for validation increasing rule compliance (R$\uparrow$).
The latter corresponding to attacking the negated rule set, where an inversed effect on any \ConPAtt{}-specific improvements is expected.
This allows to assess whether \ConPAtt{} benefits arise from training on perturbed samples in general, or whether truly semantic improvements are introduced by the attack formulation.
The full rule set is provided in the appendix (\cref{app:cv-setup}).

\paragraph{Adversarial Retraining and Evaluation Protocol}
To compare standard adversarial attacks with \ConPAtt{}, we retrain models on a mixture of clean and adversarial examples generated by each attack, while keeping the concept classifiers fixed.
Evaluation is conducted along two axes.
First, we report F1 score on the original data to assess generalization after retraining.
Second, we evaluate retraining robustness by testing adversarial examples from one attack against models retrained with another.

While rule compliance (i.e., the fraction of satisfied rules across samples and rules) is useful to characterize the effect of these attacks on semantic consistency, it does not necessarily reflect whether rules are correctly causally used for prediction.
In particular, high rule compliance can arise from spurious correlations (e.g., in biased settings where concepts of interest are never predicted, only a biased feature). Hence, it is not used as a primary evaluation metric, but reported only to verify that the generated adversarial examples consistently increase or decrease rule satisfaction as intended.
The complete analysis of this rule compliance is included in the Appendix (\cref{app:experimental-results}).

\subsection{Results}
\label{subsec:results}

\paragraph{Generalization after Retraining}
\input{tables/f1_table-original}
\Cref{tab:f1-result-original} reports F1 scores on original data after adversarial retraining.
Across all datasets and architectures, retraining consistently improves performance compared to the original models, excluding VGG-16 on MNIST.
Comparing attack strategies, R$\downarrow$ and PGD achieve comparable F1 scores, while R$\uparrow$ performs slightly but consistently lower. This shows that \ConPAtt{} benefits go beyond the ones from generally training on perturbed samples.
In the biased setting, R$\downarrow$ yields a significant improvement over PGD on the larger models, suggesting that targeting rule violations can mitigate shortcut learning effects.

We observe increased variance for VGG-16 on MNIST, where several training runs for R$\downarrow$ and R$\uparrow$ do not converge.
This effect is not present for PGD-based retraining and results in higher standard deviations and lower mean for these configurations.
Apart from that, maintained F1 scores after adversarial training are competitive to vanilla PGD.

\paragraph{Robustness to PGD Adversarial Examples}
\input{tables/f1_table-AE}
\Cref{tab:f1-result-ae} evaluates robustness by testing PGD adversarial examples across differently retrained models.
As expected, PGD adversarial examples achieve very low F1 scores on the original (non-retrained) models, in some cases achieving a 100\% attack success rate.
However, models retrained with R$\downarrow$ and R$\uparrow$ generalize well to previously unseen PGD adversarial examples, achieving substantially higher F1 scores than models without robustification.
While PGD-retrained models perform best on PGD adversarial examples, 
R$\downarrow$ achieves mostly competitive results.
This indicates that \ConPAtt{}-based retraining improves robustness beyond its own attack distribution.

\begin{figure}[htbp]
    \centering

    \begin{subfigure}[b]{0.32\linewidth}
    \centering
        \includegraphics[width=.5\linewidth]{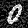}
        \caption{PGD}
    \end{subfigure}
    \hfill
    \begin{subfigure}[b]{0.32\linewidth}
    \centering
        \includegraphics[width=.5\linewidth]{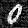}
        \caption{R$\downarrow$}
    \end{subfigure}
    \hfill
    \begin{subfigure}[b]{0.32\linewidth}
    \centering
        \includegraphics[width=.5\linewidth]{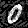}
        \caption{R$\uparrow$}
    \end{subfigure}

    \caption{Adversarial examples generated by PGD, R$\downarrow$, and R$\uparrow$ for the same random selected input sample.}
    \label{fig:ae-comparison}
\end{figure}

\paragraph{Semantic Properties of Adversarial Examples}

We also investigate 
rule compliance of adversarial examples generated by different attack strategies (detailed results in \cref{app:experimental-results}).
R$\downarrow$ produces adversarial examples that cause significantly lower rule compliance than PGD on vanilla models, whereas R$\uparrow$ increases rule compliance beyond both PGD and the original samples.
This demonstrates that, despite relying on similar optimization procedures, the two \ConPAtt{} variants generate semantically distinct adversarial examples, as can also be seen in \Cref{fig:ae-comparison}.
In particular, they systematically manipulate rule satisfaction rather than producing arbitrary perturbations, highlighting a key difference to the standard AA formulation.

\subsection{Extension to Large Language Models}
\label{subsec:llm-evaluation}

To assess whether \ConPAtt{} extends beyond image classifiers, we instantiate the same framework on a code understanding task.
Thus, a Qwen-3.5 base model~\cite{qwen3.5} is fine-tuned for binary vulnerability classification on a public  dataset~\cite{DBLP:conf/nips/ZhouLSD019}, and post-hoc probes are attached to its hidden layers to predict five memory-safety concepts.
The concept-based property used as attack target is the standard sufficient condition for the presence of a memory safety bug,
$\ConceptTerm{op} \wedge \left(\ConceptTerm{size} \vee \neg\ConceptTerm{bounds} \vee \ConceptTerm{ret\_unchk} \vee \ConceptTerm{alloc\_unchk}\right) \implies \ConceptTerm{vuln}$, evaluated under the product t-norm.
The entire experimental setup is described in detail in \Cref{app:llm-setup} of the Appendix.

Adversarial examples are created by a series of semantics-preserving transformations of the source code~\cite{DBLP:conf/uss/RisseB24} while retaining those samples that either flip the model's task output (standard AA) or flip the satisfaction value of the concept-based property (\ConPAtt{}).
The trained baseline is then fine-tuned on each adversarial set, mirroring the retraining protocol from \Cref{subsec:experimental_setup}.
For robustness, five seeds are evaluated per configuration. \Cref{tab:llm-results} reports mean and standard deviation on the held-out test set.

The pattern is consistent with the vision results.
Standard adversarial fine-tuning degrades clean F1 by $0.15$ on average, whereas \ConPAtt{} retraining preserves clean F1 and yields an improved score.
\ConPAtt{} also reduces the standard deviation across all metrics relative to the baseline, suggesting that targeting concept-based properties acts as a regularizer on the model's internal representations.

\begin{table}[tbh]
\centering
\caption{Test metrics for the baseline LLM and its two adversarially fine-tuned variants on the code vulnerability detection task. Mean $\pm$ std over five seeds.}
\label{tab:llm-results}
\begin{tabular}{lcccc}
\toprule
                   & Acc.\ & Prec.\ & Rec.\ & F1  \\
\midrule
Baseline           & $.649_{\pm.008}$ & $.658_{\pm.038}$ & $.514_{\pm.094}$ & $.569_{\pm.046}$  \\
Std.\ AA           & $.589_{\pm.037}$ & $.683_{\pm.140}$ & $.407_{\pm.284}$ & $.423_{\pm.167}$  \\
\ConPAtt{}         & $.650_{\pm.003}$ & $.645_{\pm.033}$ & $.551_{\pm.091}$ & $.587_{\pm.040}$  \\
\bottomrule
\end{tabular}
\end{table}

\section{Discussion}\label{sec:discussion}

Regarding validation of our approach, the presented results highlight that \ConPAtt{} generates adversarial examples in a structurally different way compared to standard attacks such as PGD.
The attack on rule compliance leads to systematic changes in semantic properties, as reflected by the distinct behavior of R$\downarrow$ and R$\uparrow$.
This indicates that adversarial perturbations can be guided by high-level structure rather than only low-level input variations.
These observations are consistent with the \Cref{thm:reduced-search-space}, which shows the restriction of the adversarial search space.
This restriction may bias the optimization towards semantically meaningful perturbations, providing a possible explanation for the qualitative differences observed in the generated adversarial examples.

Furthermore, regarding our strong hypothesis (2), in the biased setting we observe that R$\downarrow$ outperforms PGD under adversarial retraining.
It is more effective in mitigating semantic rule violations like the considered spurious correlations, even for maliciously crafted, easy-to-attack biases.
%
For our choice of small spurious features, 
we find that standard attacks such as PGD can also mitigate shortcut behavior to a certain extent. 
Whether similar behavior holds in more complex or realistic bias scenarios or whether \ConPAtt{} here increases the margin remains an interesting open question.

The evaluation on PGD adversarial examples further shows that models retrained with \ConPAtt{} generalize to unseen attacks.
Since these PGD adversarial examples are not observed during training, this proving our hypothesis (1):  \ConPAtt{}-based retraining can indirectly mitigate vulnerabilities that are also exploited by standard attacks.

However, our experiments also reveal limitations of rule compliance as an evaluation metric.
While it captures correlations between predictions and predefined rules, it does not necessarily reflect whether these rules are causally used by the model.
Unused rules can trivially receive optimal rule compliance values if the left hand of the rule is constantly false (none of the features detected).
This becomes apparent in biased settings, where high rule compliance can coincide with shortcut-based decision making.
As a result, novel metrics for assessing the correctness of model reasoning are needed.

Future work should further investigate and optimized convergence of \ConPAtt{} formulations, to avoid the instability we observed for some architectures (in particular VGG-16 on MNIST).




\section{Conclusion}

In this paper, we introduce \ConPAtt{}, a novel generalized adversarial attack goal: Instead of targeting a change in (respectively falsification of) the output class, our attacks aim to falsify the compliance of the NN with prior symbolic knowledge on sufficient indicators for an output class.
We showed theoretically that standard AAs arise as a special case of our formulation and that concept-based properties can restrict the adversarial search space with increasing numbers of concepts.
Beyond these theoretical properties, our experiments demonstrate that \ConPAtt{} generates semantically structured adversarial examples that differ qualitatively from standard attacks such as PGD.
Furthermore, adversarial retraining with \ConPAtt{} improves robustness and shows promising behaviour in biased settings, while also generalizing beyond the vision domain.
Overall, our results suggest that \ConPAtt{} provide a promising direction for studying robustness, semantic consistency, and the interaction between symbolic structure and adversarial learning.

%% file: tables/f1_table-original.tex
\begin{table}
    \caption{
F1 scores \textbf{on original data} after adversarial retraining.
Each model is retrained using adversarial examples generated by the corresponding method.
Values are reported as mean $\pm$ standard deviation across five runs.
R$\downarrow$ achieves performance comparable to PGD, while the wrong attack direction R$\uparrow$ is consistently lower.
In the biased setting, R$\downarrow$ yields significant improvements over PGD for the majority of models.
}
    \centering
    \begin{tabular}{llcccc}
        \toprule
        \multicolumn{2}{c}{Backbone} & No retraining & PGD & R$\downarrow$ & \textcolor{gray}{R$\uparrow$ } \\
        \midrule
        \multirow{5}{*}{\rotatebox[origin=c]{90}{MNIST}}
        & LENET-5 & $.988_{\pm.003}$ & $.997_{\pm.001}$ & $.998_{\pm.000}$ & \textcolor{gray}{ $.994_{\pm.001}$ } \\
        & M3 & $.992_{\pm.001}$ & $.998_{\pm.001}$ & $.998_{\pm.001}$ & \textcolor{gray}{ $.992_{\pm.002}$ } \\
        & ResNet-8 & $.986_{\pm.005}$ & $.998_{\pm.000}$ & $.998_{\pm.000}$ & \textcolor{gray}{ $.994_{\pm.001}$ } \\
        & ViT & $.916_{\pm.012}$ & $.940_{\pm.011}$ & $.937_{\pm.010}$ & \textcolor{gray}{ $.937_{\pm.015}$ } \\
        & VGG-16 & $.992_{\pm.001}$ & $.986_{\pm.006}$ & $.913_{\pm.253}$ & \textcolor{gray}{ $.853_{\pm.346}$ } \\
        \midrule
        \multirow{5}{*}{\rotatebox[origin=c]{90}{GTSRB}}
        & LENET-5 & $.968_{\pm.006}$ & $.973_{\pm.001}$ & $.975_{\pm.001}$ & \textcolor{gray}{ $.964_{\pm.004}$ } \\
        & M3 & $.985_{\pm.016}$ & $.993_{\pm.001}$ & $.993_{\pm.001}$ & \textcolor{gray}{ $.990_{\pm.002}$ } \\
        & ResNet-8 & $.831_{\pm.114}$ & $.907_{\pm.013}$ & $.927_{\pm.022}$ & \textcolor{gray}{ $.919_{\pm.020}$ } \\
        & ViT & $.917_{\pm.007}$ & $.932_{\pm.007}$ & $.939_{\pm.004}$ & \textcolor{gray}{ $.930_{\pm.006}$ } \\
        & VGG-16 & $.880_{\pm.001}$ & $.990_{\pm.001}$ & $.992_{\pm.004}$ & \textcolor{gray}{ $.992_{\pm.002}$ } \\
        \midrule
        \multirow{5}{*}{\rotatebox[origin=c]{90}{Biased}}
        & LENET-5 & $.522_{\pm.053}$ & $.965_{\pm.007}$ & $.965_{\pm.005}$ & \textcolor{gray}{ $.964_{\pm.004}$ } \\
        & M3 & $.989_{\pm.002}$ & $.989_{\pm.001}$ & $.989_{\pm.004}$ & \textcolor{gray}{ $.989_{\pm.002}$ } \\
        & ResNet-8 & $.296_{\pm.035}$ & $.861_{\pm.050}$ & $.881_{\pm.041}$ & \textcolor{gray}{ $.862_{\pm.050}$ } \\
        & ViT & $.132_{\pm.025}$ & $.926_{\pm.008}$ & $.939_{\pm.007}$ & \textcolor{gray}{ $.925_{\pm.011}$ } \\
        & VGG-16 & $.689_{\pm.002}$ & $.988_{\pm.010}$ & $.993_{\pm.002}$ & \textcolor{gray}{  $.991_{\pm.003}$ } \\
        \bottomrule
    \end{tabular}

    \label{tab:f1-result-original}
\end{table}

%% file: tables/f1_table-AE.tex
\begin{table}
    \caption{
    F1 scores \textbf{on PGD adversarial examples} evaluated across differently retrained models.
    PGD adversarial examples strongly degrade performance on non-retrained models, indicating high attack success.
    Models retrained with R$\downarrow$ and R$\uparrow$ generalize well to unseen PGD adversarial examples, although PGD-retrained models perform best due to exposure during training.
    Values are reported as mean $\pm$ standard deviation across five runs.
    }
    \centering
    \begin{tabular}{llcccc}
        \toprule
        \multicolumn{2}{c}{Backbone}  & No retraining & PGD & R$\downarrow$ & \textcolor{gray}{R$\uparrow$} \\
        \midrule
        \multirow{5}{*}{\rotatebox[origin=c]{90}{MNIST}}
        & LENET-5 & $.102_{\pm.084}$ & $.997_{\pm.001}$ & $.938_{\pm.033}$ & \textcolor{gray}{ $.919_{\pm.025}$ } \\
        & M3 & $.373_{\pm.114}$ & $.998_{\pm.001}$ & $.995_{\pm.001}$ & \textcolor{gray}{ $.987_{\pm.004}$ } \\
        & ResNet-8 & $.110_{\pm.013}$ & $.998_{\pm.000}$ & $.994_{\pm.002}$ & \textcolor{gray}{ $.985_{\pm.005}$ } \\
        & ViT & $.430_{\pm.102}$ & $.927_{\pm.012}$ & $.910_{\pm.014}$ & \textcolor{gray}{ $.923_{\pm.016}$ } \\
        & VGG-16 & $.000_{\pm.000}$ & $.984_{\pm.008}$ & $.909_{\pm.252}$ & \textcolor{gray}{ $.850_{\pm.345}$ } \\
        \midrule
        \multirow{5}{*}{\rotatebox[origin=c]{90}{GTSRB}}
        & LENET-5 & $.057_{\pm.012}$ & $.906_{\pm.003}$ & $.822_{\pm.009}$ & \textcolor{gray}{ $.822_{\pm.011}$ } \\
        & M3 & $.114_{\pm.024}$ & $.935_{\pm.007}$ & $.788_{\pm.025}$ & \textcolor{gray}{ $.849_{\pm.008}$ } \\
        & ResNet-8 & $.003_{\pm.002}$ & $.667_{\pm.041}$ & $.664_{\pm.037}$ & \textcolor{gray}{ $.609_{\pm.076}$ } \\
        & ViT & $.000_{\pm.000}$ & $.866_{\pm.023}$ & $.689_{\pm.029}$ & \textcolor{gray}{ $.637_{\pm.056}$ } \\
        & VGG-16 & $.000_{\pm.000}$ & $.943_{\pm.009}$ & $.938_{\pm.013}$ & \textcolor{gray}{ $.909_{\pm.020}$ } \\
        \midrule
        \multirow{5}{*}{\rotatebox[origin=c]{90}{Biased}}
        & LENET-5 & $.012_{\pm.005}$ & $.890_{\pm.014}$ & $.849_{\pm.010}$ & \textcolor{gray}{ $.861_{\pm.006}$ } \\
        & M3 & $.138_{\pm.092}$ & $.909_{\pm.004}$ & $.794_{\pm.040}$ & \textcolor{gray}{ $.832_{\pm.017}$ } \\
        & ResNet-8 & $.009_{\pm.006}$ & $.659_{\pm.064}$ & $.645_{\pm.064}$ & \textcolor{gray}{ $.616_{\pm.078}$ } \\
        & ViT & $.000_{\pm.000}$ & $.885_{\pm.012}$ & $.676_{\pm.047}$ & \textcolor{gray}{ $.591_{\pm.101}$ } \\
        & VGG-16 & $.000_{\pm.000}$ & $.943_{\pm.014}$ & $.937_{\pm.014}$ & \textcolor{gray}{ $.929_{\pm.011}$ } \\
        \bottomrule
    \end{tabular}

    \label{tab:f1-result-ae}
\end{table}

%% file: appendix.tex
\section{Concept Datasets}
\label{app:concept-datasets}

\subsection{MNIST Concept Dataset}
\label{app:mnist-concept-dataset}

We identified six concepts present in the MNIST dataset: lines, curves, circles, semicircles, right angles, and crossings. For each concept, we generated parameterized variations on a high-resolution canvas, applied global transformations, and finally downsampled the images.

\begin{enumerate}
  \item \textbf{drawing pipeline.}  
    Each concept is rendered on a $112\times112$ pixel canvas, then downsampled to $28\times28$ pixel with anti-aliasing (Gaussian smoothing). All drawing parameters (e.g.\ lengths, radii) are specified in the final $28\times28$ pixel grid.

  \item \textbf{Class variability}  
    For each concept we sampled drawing parameters:
    \begin{itemize}
      \item \textbf{Lines:} lengths of 8, 10, 12, 14, 16, and 18 pixel.
      \item \textbf{Curves:} arc angles $\{60^\circ,75^\circ,90^\circ,105^\circ,120^\circ\}$.
      \item \textbf{Circles:} radii of 4, 6, 8, 10, and 12 pixel with aspect ratios $\{0.6,0.8,1.0,1.2,1.4\}$.
      \item \textbf{Semicircles:} spans $\{150^\circ,180^\circ,300^\circ\}$.
      \item \textbf{Right angles:} lengths of 6, 8, and 10 pixel with deviations $\{80^\circ,90^\circ,100^\circ\}$.
      \item \textbf{Crossings:} lengths of 6, 8, and 10 pixel with intersection angles $\{80^\circ,90^\circ,100^\circ\}$.
    \end{itemize}
  \item \textbf{Global transformations.}  
    On each instance we applied:
    \begin{itemize}
      \item rotations from $0^\circ$ to $160^\circ$ in $20^\circ$ steps,
      \item translations of $\{-5,-2.5,0,2.5,5\}$px in both $x$ and $y$,
      \item constant scale factor $s=1.0$.
    \end{itemize}

  \item \textbf{Dataset summary.}  
    In total, we generated 1{,}500 labeled images with size $28\times28$px for each concept class.
    These samples are partitioned into fixed training ($72.25\%$), validation ($12.25\%$) and test datasets ($15\%$).
\end{enumerate}
    
\subsection{GTSRB Concept Dataset}
\label{app:gtsrb-concept-dataset}

We identified 3 groups of concepts for traffic signs --- shape, color, and symbol.
The concepts for the shape are triangle, inverted triangle, circle, octagon, and diamond; for the color red, blue, yellow with a white border ($\ConceptTerm{white\_yellow}$), white with a red border ($\ConceptTerm{white\_red}$) and white with black symbols ($\ConceptTerm{white\_black}$).
The symbols basically correspond to the usual figures or lettering on static traffic signs, such as "STOP" or the speed limit as a lettering, but also signs such as "right/left curve", "snowflake", "roadworks" or two cars next to each other (general "no overtaking") but no place or country names.
Each sample of the GTSRB concept dataset is a combination of three concepts --- one from each concept group --- with a random background.
This dataset therefore also contains completely unrealistic images, such as a diamond-shaped, blue stop sign and a triangular, yellow 100 km/h speed limit sign (\Cref{fig:gtsrb-concept-dataset-samples}).
As a result, concept recognition cannot utilize the "clever Hans" effect.
With realistic images, there would be, for example, only one sign class which is octagonal and has the word "STOP" printed on it – these two concepts would therefore be indistinguishable using realistic images.

In total, we generated 7{,}227 labeled RGB images with size $64\times64$px. 
Each symbol concept is contained at least 15 times, and  all shape and color concepts are included to the same extent.
These samples are partitioned into fixed training ($60 \%$), validation ($20 \%$) and test datasets ($20 \%$).

\begin{figure}
    \centering
    \includegraphics{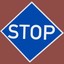}
    \includegraphics{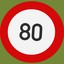}
    \includegraphics{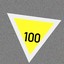}
    \caption{Samples of the GTSRB concept dataset.}
    \label{fig:gtsrb-concept-dataset-samples}
\end{figure}

\section{Experimental Setup Details}
\label{app:cv-setup}

\paragraph{Biased Dataset}
Based on GTSRB, we add a coloured rectangle of 5 $\times$ 5 pixels on the left side of the image.
The rectangle position is fixed, and the colour is derived from the label class.
Each class has its own unique colour.
The image size is standardized to 64 $\times$ 64 pixels before adding the bias.

\paragraph{Models}
We use the following architectures:
\begin{itemize}
 \item LeNet-5~\cite{lecun_gradient-based_1998},
 \item M3~\cite{an_ensemble_2020},
 \item ResNet-8~\cite{he_deep_2016},
 \item VGG-16~\cite{simonyan_very_2015} (pretrained on ImageNet),
 \item Vision Transformer (ViT)~\cite{dosovitskiy_image_2021} with  8 attention heads, hidden dimension 256, dropout rate 0.1, and patch size 4.
\end{itemize}

\paragraph{Model Training}
All models are trained using the Adam optimizer with default parameters ($\beta_1=0.9$, $\beta_2=0.999$).
The learning rate is set to $10^{-3}$ for MNIST and $10^{-4}$ for GTSRB.
A batch size of 64 is used throughout all experiments.

Early stopping is applied based on validation loss with a threshold of $10^{-3}$.
Each model is trained with five random seeds.

Inputs are resized using bilinear interpolation if required by the architecture.
Standard data augmentation is applied during training, including random rotations, affine transformations with translation, and random resized crops.

Cross-entropy loss is used for all classification tasks.
Temperature scaling is used for calibration.

\paragraph{Concept Mining}
Concept classifiers are trained on intermediate feature representations of each model using three random seeds.
We consider up to three layers per architecture to capture different levels of abstraction.

For each concept, five linear classifiers with sigmoid activation are trained using balanced binary cross-entropy loss.
Predictions are calibrated via temperature scaling.
The final concept classifier is obtained by averaging weights and biases across the ensemble.

Concept datasets are augmented and partially mixed with original task data to mitigate domain shift.

Concept definitions are task-specific.  
\begin{itemize}
    \item MNIST:
    \begin{itemize}
        \item  $\ConceptTerm{curve}$
        \item  $\ConceptTerm{semicircle} \subset \ConceptTerm{curve}$
        \item  $\ConceptTerm{circle} \subset \ConceptTerm{semicircle}$
        \item  $\ConceptTerm{line}$
        \item  $\ConceptTerm{right\_angle}  \subset  \ConceptTerm{line}$
        \item  $\ConceptTerm{crossing} \subset  \ConceptTerm{right\_angle}$
    \end{itemize}
    \item GTSRB:
    \begin{itemize}
        \item $\ConceptTerm{octogonal}$
        \item $\ConceptTerm{circle}$
        \item $\ConceptTerm{diamond}$
        \item $\ConceptTerm{inverted\_triangle}$
        \item $\ConceptTerm{red}$
        \item $\ConceptTerm{white\_yellow}$
        \item $\ConceptTerm{white\_red}$
        \item $\ConceptTerm{stop\_symbol}$
        \item $\ConceptTerm{symbol\_30}$
        \item $\ConceptTerm{symbol\_50}$
        \item $\ConceptTerm{symbol\_120}$
        \item $\ConceptTerm{symbol\_no\_overtaking\_truck}$
        \item $\ConceptTerm{symbol\_no\_overtaking\_pkw}$
    \end{itemize}
\end{itemize}

\paragraph{Adversarial Attacks}
Adversarial examples are generated using the \textit{torchattack} library.
We use a standard PGD implementation with $\epsilon = 0.3$ (determined via sweep on MNIST).
Unless stated otherwise, default parameters are used.

\paragraph{Rules}
\begin{itemize}
    \item MNIST:
    \begin{itemize}
        \item  $\ConceptTerm{zero}\implies\ConceptTerm{circle}$
        \item  $\ConceptTerm{one}\implies\ConceptTerm{line}$
        \item  $\ConceptTerm{two}\implies\ConceptTerm{line} \land \ConceptTerm{semicircle}$
        \item  $\ConceptTerm{two}\implies\ConceptTerm{line} \land \ConceptTerm{curve}$
        \item  $\ConceptTerm{three}\implies\ConceptTerm{semicircle}$
        \item  $\ConceptTerm{four}\implies\ConceptTerm{line} \land \ConceptTerm{crossing}$
        \item  $\ConceptTerm{five}\implies\ConceptTerm{line} \land \ConceptTerm{right\_angle} \land \ConceptTerm{semicircle}$
        \item  $\ConceptTerm{five}\implies\ConceptTerm{curve} \land \ConceptTerm{right\_angle} \land \ConceptTerm{semicircle}$
        \item  $\ConceptTerm{six}\implies\ConceptTerm{curve} \land \ConceptTerm{circle}$
        \item  $\ConceptTerm{seven}\implies\ConceptTerm{line} \land \ConceptTerm{right\_angle}$
        \item  $\ConceptTerm{seven}\implies\ConceptTerm{line} \land \ConceptTerm{right\_angle} \land \ConceptTerm{crossing}$
        \item  $\ConceptTerm{eight}\implies\ConceptTerm{semicircle}$
        \item  $\ConceptTerm{eight}\implies\ConceptTerm{semicircle} \land \ConceptTerm{crossing}$
        \item  $\ConceptTerm{nine}\implies\ConceptTerm{curve} \land \ConceptTerm{circle}$
        \item  $\ConceptTerm{curve}\implies\ConceptTerm{zero}\lor \ConceptTerm{two}\lor \ConceptTerm{three}\lor \ConceptTerm{five} \lor \ConceptTerm{six} \lor \ConceptTerm{eight} \lor \ConceptTerm{nine}$
        \item  $\ConceptTerm{semicircle}\implies\ConceptTerm{zero}\lor \ConceptTerm{two}\lor \ConceptTerm{three}\lor \ConceptTerm{five} \lor \ConceptTerm{six} \lor \ConceptTerm{eight} \lor \ConceptTerm{nine}$
        \item  $\ConceptTerm{circle}\implies\ConceptTerm{zero} \lor \ConceptTerm{six} \lor \ConceptTerm{eight} \lor \ConceptTerm{nine}$
        \item  $\ConceptTerm{line}\implies\ConceptTerm{one} \lor \ConceptTerm{two} \lor \ConceptTerm{four} \lor \ConceptTerm{five} \lor \ConceptTerm{seven}$
        \item  $\ConceptTerm{right\_angle}\implies\ConceptTerm{four} \lor \ConceptTerm{five} \lor \ConceptTerm{seven}$
        \item  $\ConceptTerm{crossing}\implies \ConceptTerm{four} \lor  \ConceptTerm{seven}$
        \item  $\ConceptTerm{crossing}\implies \ConceptTerm{four} \lor  \ConceptTerm{seven} \lor  \ConceptTerm{eight}$
    \end{itemize}
    \item GTSRB:
    \begin{itemize}
        \item $\ConceptTerm{stop\_sign}\Longleftrightarrow\ConceptTerm{octogonal}\land\ConceptTerm{red}\land\ConceptTerm{stop\_symbol}$
        \item $\ConceptTerm{speed\_limit\_30}\Longleftrightarrow\ConceptTerm{circle}\land\ConceptTerm{white\_red}\land\ConceptTerm{symbol\_30}$
        \item $\ConceptTerm{speed\_limit\_50}\Longleftrightarrow\ConceptTerm{circle}\land\ConceptTerm{white\_red}\land\ConceptTerm{symbol\_50}$
        \item $\ConceptTerm{speed\_limit\_120}\Longleftrightarrow\ConceptTerm{circle}\land\ConceptTerm{white\_red}\land\ConceptTerm{symbol\_120}$
        \item $\ConceptTerm{no\_overtaking\_trucks}\Longleftrightarrow\ConceptTerm{circle}\land\ConceptTerm{white\_red}\land\ConceptTerm{symbol\_no\_overtaking\_truck}$
        \item $\ConceptTerm{no\_overtaking}\Longleftrightarrow\ConceptTerm{circle}\land\ConceptTerm{white\_red}\land\ConceptTerm{symbol\_no\_overtaking\_pkw}$
        \item $\ConceptTerm{priority\_road}\Longleftrightarrow\ConceptTerm{diamond}\land\ConceptTerm{white\_yellow}$
        \item $\ConceptTerm{give\_way}\Longleftrightarrow\ConceptTerm{inverted\_triangle}\land\ConceptTerm{white\_red}$
        \item $\ConceptTerm{no\_entry}\Longleftrightarrow\ConceptTerm{circle}\land\ConceptTerm{white\_red}$
    \end{itemize}
\end{itemize}

\paragraph{Adversarial Retraining}
Models are retrained using a 1:1 mixture of clean and adversarial samples.

\paragraph{Reproducibility}
All experiments are conducted with fixed random seeds across libraries to ensure reproducibility using an NVIDIA H100 NVL GPU and a containerized PyTorch environment (version 24.11).

\section{CV Additional Experimental Results}
\label{app:experimental-results}

\paragraph{Attack Success Rate}
Detailed results on the attack success rate of the different attack approaches are reported in \cref{tab:asr-result}.
\input{tables/asr}

\paragraph{Different Adversarial Attacks}
To assess whether the observed behaviour depends on the particular adversarial attack, we additionally evaluated FGSM and VMI-FGSM besides the PGD attack used throughout the main paper.
The results are summarized in \Cref{tab:f1-result-AA-comparison-original,tab:f1-result-AA-comparison-AE}.
\Cref{tab:f1-result-AA-comparison-original} shows that the differences between the standard attacks and the corresponding \ConPAtt{} variants remain comparable across all considered attack methods and model architectures.
Furthermore, \Cref{tab:f1-result-AA-comparison-AE} demonstrates that adversarial examples generated by the standard attacks continue to result in few misclassifications on adversarially retrained models to a similar extent for FGSM and VMI-FGSM as observed for PGD.
Overall, these results indicate that the conclusions drawn in the main paper are not specific to PGD but generalize across different attack methods.
\input{tables/f1_table_AA_comparison}

\paragraph{Rule Compliance}
Detailed results regarding the rule compliance metric are provided in \cref{tab:rule-compliance-result} (on original data), \cref{tab:rule-compliance-result-rct-ae} (on attacked samples) and \cref{tab:rule-compliance-result-irct-ae} (on samples attacked with R$\uparrow$).

\paragraph{Concept Performance}
Since \ConPAtt{} targets the concept predictions, the quality of the employed concept probes is an important prerequisite.
Poor concept prediction performance would imply that the attacks are applied to unreliable concept representations rather than the intended semantic concepts.
Therefore, we report the concept prediction performance both before and after retraining.
While the pre-retraining performance validates the concept probes used by \ConPAtt{}, the post-retraining performance is reported for transparency only, as it does not affect the task performance of the adversarially retrained models.
The corresponding results are given for MNIST in \Cref{tab:f1-result-original}, for GTSRB in \Cref{tab:f1-gtsrb-result-original,tab:f1-gtsrb-result-summary}, and for the biased GTSRB variant in \Cref{tab:f1-biased-gtsrb-result-original,tab:f1-biased-gtsrb-result-summary}.
For GTSRB, we report both the performance of the individual concepts and aggregated results for the colour, shape, and symbol concept groups.
Across the evaluated architectures, the observed changes after retraining depend on the model and dataset.
VGG16 exhibits the largest decrease in concept prediction performance, particularly for symbol concepts, whereas the remaining architectures show only moderate changes.
The biased-trained GTSRB model generally exhibits larger performance differences than the original GTSRB model.

\input{tables/concept_progression/mnist_f1}
\input{tables/concept_progression/gtsrb_f1}
\input{tables/concept_progression/biased_gtsrb_f1}

\begin{figure}
    \centering
    \includegraphics[width=.9\linewidth]{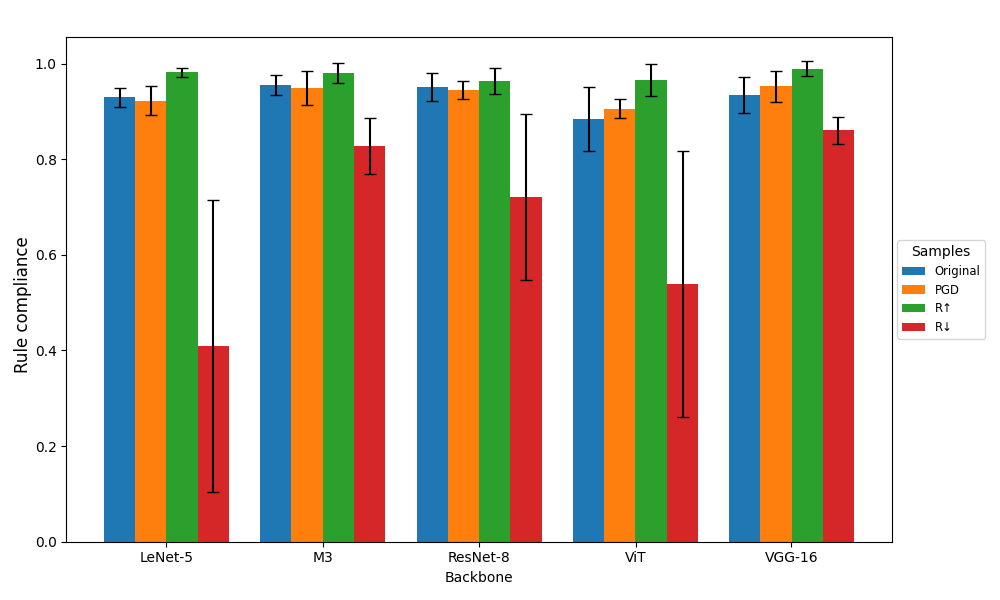}
    \caption{Rule compliance of each architecture, differentiated by samples type. R$\downarrow$ generates, on average, adversarial examples with the lowest rule compliance, and R$\uparrow$ those with the highest.}
    \label{fig:rule-compliance}
\end{figure}

\input{tables/rule_table-original}
\input{tables/rule_table-RCT-AE}

\input{tables/rule_table-IRCT-AE}

\section{LLM Vulnerability Detection: Setup}
\label{app:llm-setup}

This section details the experimental configuration of the LLM extension reported in \Cref{subsec:llm-evaluation}.

\textbf{Dataset and Task.}
We use the Devign benchmark \cite{DBLP:conf/nips/ZhouLSD019}, which contains functions extracted from large-scale open-source C projects, each labeled as vulnerable or safe.
The task is binary classification at the function level.
Each run draws an independent train/validation split.

\textbf{Model and Training.}
The classifier is a $0.8$\,B-parameter Qwen-3.5 base model \cite{qwen3.5} extended with rank-$16$ LoRA adapters \cite{DBLP:conf/iclr/HuSWALWWC22} on the query and value projections of every transformer block.
The model is prompted with the candidate function and asked to output a single token (\textsc{Yes}/\textsc{No}). The cross-entropy loss is computed only on the two-token logit subset.
The baseline is trained for three epochs at learning rate $2{\times}10^{-4}$ with $10\%$ linear warmup and linear decay, weight decay $0.01$, and effective batch size $16$.
Adversarial fine-tuning continues from the baseline LoRA for two epochs at $5{\times}10^{-5}$ on the adversarial subset only.

\textbf{Concepts and Probes.}
The five memory-safety concepts $\ConceptTerm{op}$ (unsafe memory operation), $\ConceptTerm{size}$ (risky size computation), $\ConceptTerm{bounds}$ (bounds check), $\ConceptTerm{ret\_unchk}$ (unchecked return value of a side-effecting call), and $\ConceptTerm{alloc\_unchk}$ (unchecked allocation) are extracted from the abstract syntax tree of the source code using tree-sitter \cite{tree-sitter}.
A linear probe is trained per concept on mean-pooled hidden states of $3{,}000$ training samples following the post-hoc protocol of \Cref{subsec:fuzzy-logic}. Ultimately the probe with the highest validation AUROC is selected.
All five probes have an AUROC between $0.918$ and $0.951$ across runs.

\textbf{Code Transformations.}
We use the eight semantics-preserving code transformations from Risse \emph{et al}~\cite{DBLP:conf/uss/RisseB24}: parameter rename, function rename, dead-code insertion, comment insertion, whitespace insertion, helper-function wrapping, void-call insertion, and comment removal.
For each training sample we sample random transformation chains of length one to four.

\textbf{Adversarial Generation.}
For standard adversarial examples a candidate is retained as \emph{output-targeted} if it flips the trained baseline's $\textsc{Yes}/\textsc{No}$ prediction; as \emph{\ConPAtt{}-targeted} if it flips the binary value of the concept-based property. 
Output-targeted generation yields on average $27{,}032$ candidates per run (range $18{,}017$ to $40{,}064$); \ConPAtt{}-targeted yields $2{,}536$ (range $2{,}475$ to $2{,}623$).
For the comparisons in \Cref{tab:llm-results}, each adversarial set is uniformly subsampled to the smaller of the two counts so that fine-tuning sees the same number of examples in both conditions.

%% file: tables/asr.tex
\begin{table}[t]
    \centering
    \caption{Attack success rate with respect to misclassification for standard attacks and their \ConPAtt{} variants across datasets and architectures. 
    Results are reported as mean $\pm$ standard deviation over multiple runs. }
    \begin{tabular}{llccccc}
        \toprule
         & Attack & LENET-5 & M3 & ResNet8  & VGG-16 & ViT \\
        \midrule
        \multirow{3}{*}{\rotatebox[origin=c]{90}{MNIST}}
			& PGD & $.898_{\pm.084}$ & $.627_{\pm.114}$ & $.890_{\pm.013}$ & $1.0_{\pm.000}$ & $.570_{\pm.102}$ \\
			& R$\uparrow$  & $.174_{\pm.028}$ & $.413_{\pm.093}$ & $.893_{\pm.011}$ & $.970_{\pm.007}$ & $.511_{\pm.096}$ \\
			& R$\downarrow$  & $.401_{\pm.261}$ & $.381_{\pm.142}$ & $.896_{\pm.011}$ & $.861_{\pm.002}$ & $.355_{\pm.065}$ \\
        \midrule
        \multirow{3}{*}{\rotatebox[origin=c]{90}{GTSRB}}
            & PGD & $.943_{\pm.012}$ & $.886_{\pm.024}$ & $.997_{\pm.002}$ & $1.0_{\pm.000}$ & $1.0_{\pm.000}$ \\
			& R$\uparrow$ & $.865_{\pm.026}$ & $.913_{\pm.017}$ & $1.0_{\pm.000}$ & $1.0_{\pm.000}$ & $1.0_{\pm.000}$ \\
			& R$\downarrow$ & $.689_{\pm.010}$ & $.767_{\pm.033}$ & $.994_{\pm.007}$ & $.999_{\pm.000}$ & $.992_{\pm.004}$ \\
        \midrule
        \multirow{3}{*}{\rotatebox[origin=c]{90}{Biased}}
            & PGD & $.988_{\pm.005}$ & $.862_{\pm.092}$ & $.991_{\pm.006}$ & $1.0_{\pm.000}$ & $1.0_{\pm.000}$ \\ 
			& R$\uparrow$ & $.968_{\pm.006}$ & $.885_{\pm.077}$ & $.993_{\pm.008}$ & $1.0_{\pm.000}$ & $.999_{\pm.001}$ \\
			& R$\downarrow$ & $.897_{\pm.019}$ & $.709_{\pm.055}$ & $.991_{\pm.005}$ & $.999_{\pm.000}$ & $.997_{\pm.002}$ \\
        \bottomrule
    \end{tabular}

    \label{tab:asr-result}
\end{table}

%% file: tables/f1_table_AA_comparison.tex
\begin{table}
    \caption{
    F1 scores \textbf{on original data} after adversarial retraining.
    Each model is retrained using adversarial examples generated by the corresponding method.
    Values are reported as mean $\pm$ standard deviation across five runs.}
    \centering
    \begin{tabular}{lccccc}
    	\toprule
    	Retrained & LeNet-5 & M3 & ResNet-8 & VGG-16 & ViT \\
    	\midrule
    	No retraining & $.988_{\pm.003}$ & $.992_{\pm.001}$ & $.986_{\pm.005}$ & $.992_{\pm.001}$ & $.916_{\pm.012}$ \\
    	\midrule
    	PGD & $.997_{\pm.001}$ & $.998_{\pm.001}$ & $.998_{\pm.000}$ & $.986_{\pm.006}$ & $.940_{\pm.011}$ \\
    	PGD$\downarrow$ & $.998_{\pm.000}$ & $.998_{\pm.001}$ & $.998_{\pm.000}$ & $.913_{\pm.253}$ & $.937_{\pm.010}$ \\
    	\textcolor{gray}{PGD$\uparrow$} & \textcolor{gray}{$.994_{\pm.001}$} & \textcolor{gray}{$.992_{\pm.002}$} & \textcolor{gray}{$.994_{\pm.001}$} & \textcolor{gray}{$.853_{\pm.346}$} & \textcolor{gray}{$.937_{\pm.015}$} \\
    	\midrule
    	FGSM & $.991_{\pm.001}$ & $.991_{\pm.002}$ & $.989_{\pm.009}$ & $.986_{\pm.004}$ & $.939_{\pm.004}$ \\
    	FGSM$\downarrow$  & $.992_{\pm.002}$ & $.992_{\pm.002}$ & $.993_{\pm.002}$ & $.847_{\pm.344}$ & $.945_{\pm.011}$ \\
    	\midrule
    	VMI-FGSM & $.991_{\pm.002}$ & $.991_{\pm.003}$ & $.992_{\pm.002}$ & $.986_{\pm.005}$ & $.934_{\pm.012}$ \\
    	VMI-FGSM$\downarrow$  & $.993_{\pm.001}$ & $.993_{\pm.003}$ & $.994_{\pm.002}$ & $.981_{\pm.009}$ & $.942_{\pm.009}$ \\
    	\bottomrule
    \end{tabular}
    \label{tab:f1-result-AA-comparison-original}
\end{table}

\begin{table}
    \caption{
    F1 scores \textbf{on adversarial examples of PGD, FGSM and VMIFGSM} evaluated across differently retrained models.
    Values are reported as mean $\pm$ standard deviation across five runs.
    }
    \centering
    \begin{tabular}{llccccc}
        \toprule
        AE & Retrained & LeNet-5 & M3 & ResNet-8 & VGG-16 & ViT \\
        \midrule
        \multirow{4}{*}{PGD}
        & No retraining & $.102_{\pm.084}$ & $.373_{\pm.114}$ & $.110_{\pm.013}$ & $.000_{\pm.000}$ & $.430_{\pm.102}$ \\
        & PGD & $.997_{\pm.001}$ & $.998_{\pm.001}$ & $.998_{\pm.000}$ & $.984_{\pm.008}$ & $.927_{\pm.012}$ \\
        & R$\downarrow$ & $.938_{\pm.033}$ & $.995_{\pm.001}$ & $.994_{\pm.002}$ & $.909_{\pm.252}$ & $.910_{\pm.014}$ \\
        & \textcolor{gray}{R$\uparrow$} & \textcolor{gray}{$.919_{\pm.025}$} & \textcolor{gray}{$.987_{\pm.004}$} & \textcolor{gray}{$.985_{\pm.005}$} & \textcolor{gray}{$.850_{\pm.345}$} & \textcolor{gray}{$.923_{\pm.016}$} \\
        \midrule
        \multirow{3}{*}{FGSM}
        & No retraining & $.270_{\pm.071}$ & $.703_{\pm.065}$ & $.098_{\pm.000}$ & $.102_{\pm.003}$ & $.213_{\pm.061}$ \\
        & FGSM & $.989_{\pm.002}$ & $.989_{\pm.005}$ & $.990_{\pm.004}$ & $.989_{\pm.005}$ & $.916_{\pm.009}$ \\
        & FGSM$\downarrow$ & $.904_{\pm.034}$ & $.983_{\pm.005}$ & $.975_{\pm.010}$ & $.840_{\pm.341}$ & $.908_{\pm.017}$ \\
        \midrule
        \multirow{3}{*}{VMI-FGSM}
        & No retraining & $.176_{\pm.098}$ & $.300_{\pm.113}$ & $.121_{\pm.044}$ & $.024_{\pm.003}$ & $.271_{\pm.106}$ \\
        & VMI-FGSM & $.990_{\pm.005}$ & $.992_{\pm.003}$ & $.992_{\pm.002}$ & $.984_{\pm.007}$ & $.927_{\pm.013}$ \\
        & VMI-FGSM$\downarrow$ & $.937_{\pm.033}$ & $.981_{\pm.006}$ & $.978_{\pm.010}$ & $.975_{\pm.012}$ & $.910_{\pm.016}$ \\
        \bottomrule
    \end{tabular}
    
    \label{tab:f1-result-AA-comparison-AE}
\end{table}

%% file: tables/concept_progression/mnist_f1.tex
\begin{table}
    \caption{
        F1 scores \textbf{on original MNSIT concept data} before and after adversarial retraining, reported as mean $\pm$ standard deviation across fifteen runs.
    }
    \centering
    \begin{tabular}{llccccc}
        \toprule
        Concept & Retrained  & LeNet-5 & M3 & ResNet-8 & VGG-16 & ViT \\
        \midrule
        \multirow[t]{4}{*}{Circle} & No retraining & $.803_{\pm.042}$ & $.896_{\pm.014}$ & $.921_{\pm.030}$ & $.962_{\pm.002}$ & $.766_{\pm.018}$ \\
         & PGD & $.824_{\pm.040}$ & $.909_{\pm.013}$ & $.981_{\pm.005}$ & $.906_{\pm.018}$ & $.825_{\pm.018}$ \\
         & R$\downarrow$ & $.813_{\pm.043}$ & $.915_{\pm.021}$ & $.981_{\pm.004}$ & $.908_{\pm.031}$ & $.826_{\pm.020}$ \\
         & R$\uparrow$ & $.818_{\pm.038}$ & $.923_{\pm.014}$ & $.981_{\pm.004}$ & $.883_{\pm.139}$ & $.815_{\pm.030}$ \\
        \midrule
        \multirow[t]{4}{*}{Crossing} & No retraining & $.540_{\pm.031}$ & $.896_{\pm.039}$ & $.816_{\pm.044}$ & $.913_{\pm.007}$ & $.539_{\pm.050}$ \\
         & PGD & $.517_{\pm.066}$ & $.848_{\pm.046}$ & $.967_{\pm.008}$ & $.714_{\pm.074}$ & $.623_{\pm.072}$ \\
         & R$\downarrow$ & $.554_{\pm.064}$ & $.873_{\pm.041}$ & $.970_{\pm.004}$ & $.704_{\pm.105}$ & $.606_{\pm.045}$ \\
         & R$\uparrow$ & $.570_{\pm.022}$ & $.869_{\pm.040}$ & $.968_{\pm.008}$ & $.699_{\pm.197}$ & $.624_{\pm.042}$ \\
        \midrule
        \multirow[t]{4}{*}{Curve} & No retraining & $.855_{\pm.015}$ & $.949_{\pm.012}$ & $.936_{\pm.015}$ & $.983_{\pm.001}$ & $.840_{\pm.022}$ \\
         & PGD & $.853_{\pm.023}$ & $.907_{\pm.020}$ & $.966_{\pm.010}$ & $.884_{\pm.019}$ & $.849_{\pm.019}$ \\
         & R$\downarrow$ & $.839_{\pm.026}$ & $.918_{\pm.017}$ & $.967_{\pm.012}$ & $.893_{\pm.019}$ & $.854_{\pm.017}$ \\
         & R$\uparrow$ & $.853_{\pm.017}$ & $.912_{\pm.025}$ & $.968_{\pm.013}$ & $.832_{\pm.234}$ & $.858_{\pm.017}$ \\
        \midrule
        \multirow[t]{4}{*}{Line} & No retraining & $.818_{\pm.010}$ & $.895_{\pm.021}$ & $.855_{\pm.026}$ & $.944_{\pm.003}$ & $.769_{\pm.039}$ \\
         & PGD & $.821_{\pm.020}$ & $.871_{\pm.015}$ & $.944_{\pm.016}$ & $.864_{\pm.020}$ & $.817_{\pm.015}$ \\
         & R$\downarrow$ & $.818_{\pm.022}$ & $.883_{\pm.012}$ & $.941_{\pm.013}$ & $.859_{\pm.018}$ & $.806_{\pm.034}$ \\
         & R$\uparrow$ & $.832_{\pm.014}$ & $.888_{\pm.016}$ & $.942_{\pm.014}$ & $.845_{\pm.074}$ & $.821_{\pm.007}$ \\
        \midrule
        \multirow[t]{4}{*}{Right Angle} & No retraining & $.664_{\pm.014}$ & $.884_{\pm.048}$ & $.851_{\pm.032}$ & $.930_{\pm.003}$ & $.679_{\pm.027}$ \\
         & PGD & $.627_{\pm.054}$ & $.808_{\pm.046}$ & $.965_{\pm.006}$ & $.795_{\pm.024}$ & $.747_{\pm.025}$ \\
         & R$\downarrow$ & $.659_{\pm.032}$ & $.842_{\pm.041}$ & $.963_{\pm.008}$ & $.775_{\pm.053}$ & $.714_{\pm.039}$ \\
         & R$\uparrow$ & $.682_{\pm.019}$ & $.828_{\pm.039}$ & $.967_{\pm.008}$ & $.772_{\pm.093}$ & $.717_{\pm.030}$ \\
        \midrule
        \multirow[t]{4}{*}{Semi-Circle} & No retraining & $.909_{\pm.021}$ & $.943_{\pm.008}$ & $.931_{\pm.013}$ & $.955_{\pm.005}$ & $.886_{\pm.019}$ \\
         & PGD & $.909_{\pm.036}$ & $.931_{\pm.014}$ & $.963_{\pm.004}$ & $.938_{\pm.016}$ & $.914_{\pm.006}$ \\
         & R$\downarrow$ & $.896_{\pm.033}$ & $.935_{\pm.009}$ & $.965_{\pm.004}$ & $.940_{\pm.019}$ & $.903_{\pm.011}$ \\
         & R$\uparrow$ & $.898_{\pm.029}$ & $.938_{\pm.012}$ & $.963_{\pm.004}$ & $.876_{\pm.243}$ & $.909_{\pm.004}$ \\
        \cline{1-7}
        \midrule
         \multirow[t]{4}{*}{All} & No retraining & $.765_{\pm.128}$ & $.911_{\pm.037}$ & $.885_{\pm.054}$ & $.948_{\pm.023}$ & $.746_{\pm.118}$ \\
         & PGD & $.759_{\pm.145}$ & $.879_{\pm.051}$ & $.965_{\pm.014}$ & $.850_{\pm.083}$ & $.796_{\pm.098}$ \\
         & R$\downarrow$ & $.763_{\pm.125}$ & $.895_{\pm.041}$ & $.964_{\pm.015}$ & $.847_{\pm.097}$ & $.785_{\pm.103}$ \\
         & R$\uparrow$ & $.775_{\pm.117}$ & $.893_{\pm.045}$ & $.965_{\pm.015}$ & $.818_{\pm.183}$ & $.791_{\pm.098}$ \\
        \cline{1-7}
        \bottomrule
    \end{tabular}

    \label{tab:f1-mnist-result-original}
\end{table}

%% file: tables/concept_progression/gtsrb_f1.tex
\begin{table}
    \caption{
        F1 scores \textbf{on original GTSRB concept data} of each concept before and after adversarial retraining, reported as mean $\pm$ standard deviation across fifteen runs.
    }
    \centering
    \begin{tabular}{llccccc}
        \toprule
        Concept & Retrained & LeNet-5 & M3 & ResNet-8 & VGG-16 & ViT \\
        \midrule
        \multirow[t]{4}{*}{Circle} & No retraining & $.967_{\pm.019}$ & $.999_{\pm.001}$ & $1.0_{\pm.000}$ & $1.0_{\pm.000}$ & $.999_{\pm.000}$ \\
         & PGD & $.729_{\pm.173}$ & $.990_{\pm.007}$ & $.905_{\pm.144}$ & $.447_{\pm.324}$ & $.906_{\pm.052}$ \\
         & R$\downarrow$ & $.793_{\pm.038}$ & $.991_{\pm.005}$ & $.897_{\pm.149}$ & $.463_{\pm.354}$ & $.916_{\pm.059}$ \\
         & R$\uparrow$ & $.682_{\pm.176}$ & $.981_{\pm.016}$ & $.913_{\pm.155}$ & $.478_{\pm.362}$ & $.960_{\pm.027}$ \\
        \midrule
        \multirow[t]{4}{*}{Diamond} & No retraining & $.998_{\pm.001}$ & $1.0_{\pm.000}$ & $1.0_{\pm.000}$ & $1.0_{\pm.000}$ & $1.0_{\pm.000}$ \\
         & PGD & $.942_{\pm.117}$ & $1.0_{\pm.000}$ & $.921_{\pm.154}$ & $.996_{\pm.002}$ & $.987_{\pm.006}$ \\
         & R$\downarrow$ & $.986_{\pm.007}$ & $1.0_{\pm.000}$ & $.917_{\pm.160}$ & $.997_{\pm.002}$ & $.989_{\pm.011}$ \\
         & R$\uparrow$ & $.910_{\pm.123}$ & $1.0_{\pm.000}$ & $.923_{\pm.154}$ & $.999_{\pm.001}$ & $.987_{\pm.017}$ \\
        \midrule
        \multirow[t]{4}{*}{Inverted Triangle} & No retraining & $.999_{\pm.000}$ & $1.0_{\pm.000}$ & $1.0_{\pm.000}$ & $1.0_{\pm.000}$ & $1.0_{\pm.000}$ \\
         & PGD & $.961_{\pm.089}$ & $1.0_{\pm.000}$ & $.978_{\pm.046}$ & $.999_{\pm.001}$ & $.978_{\pm.017}$ \\
         & R$\downarrow$ & $.997_{\pm.002}$ & $1.0_{\pm.000}$ & $.982_{\pm.041}$ & $.997_{\pm.003}$ & $.973_{\pm.022}$ \\
         & R$\uparrow$ & $.952_{\pm.091}$ & $1.0_{\pm.000}$ & $.988_{\pm.027}$ & $.999_{\pm.002}$ & $.986_{\pm.007}$ \\
        \midrule
        \multirow[t]{4}{*}{LKW no overtaken} & No retraining & $.631_{\pm.123}$ & $.658_{\pm.256}$ & $.648_{\pm.147}$ & $.756_{\pm.062}$ & $.426_{\pm.124}$ \\
         & PGD & $.368_{\pm.139}$ & $.803_{\pm.038}$ & $.514_{\pm.144}$ & $.005_{\pm.006}$ & $.387_{\pm.150}$ \\
         & R$\downarrow$ & $.476_{\pm.071}$ & $.825_{\pm.041}$ & $.473_{\pm.133}$ & $.007_{\pm.009}$ & $.435_{\pm.167}$ \\
         & R$\uparrow$ & $.357_{\pm.125}$ & $.766_{\pm.040}$ & $.542_{\pm.137}$ & $.023_{\pm.023}$ & $.371_{\pm.156}$ \\
        \midrule
        \multirow[t]{4}{*}{Limit 120} & No retraining & $.376_{\pm.043}$ & $.392_{\pm.231}$ & $.384_{\pm.116}$ & $.644_{\pm.055}$ & $.168_{\pm.033}$ \\
         & PGD & $.085_{\pm.035}$ & $.648_{\pm.059}$ & $.252_{\pm.092}$ & $.000_{\pm.000}$ & $.163_{\pm.044}$ \\
         & R$\downarrow$ & $.091_{\pm.030}$ & $.663_{\pm.044}$ & $.232_{\pm.076}$ & $.000_{\pm.000}$ & $.170_{\pm.056}$ \\
         & R$\uparrow$ & $.095_{\pm.046}$ & $.527_{\pm.053}$ & $.275_{\pm.096}$ & $.002_{\pm.004}$ & $.145_{\pm.040}$ \\
        \midrule
        \multirow[t]{4}{*}{Limit 30} & No retraining & $.357_{\pm.032}$ & $.440_{\pm.269}$ & $.302_{\pm.116}$ & $.384_{\pm.119}$ & $.165_{\pm.042}$ \\
         & PGD & $.102_{\pm.041}$ & $.650_{\pm.036}$ & $.288_{\pm.090}$ & $.023_{\pm.051}$ & $.161_{\pm.034}$ \\
         & R$\downarrow$ & $.132_{\pm.037}$ & $.673_{\pm.057}$ & $.274_{\pm.112}$ & $.013_{\pm.027}$ & $.162_{\pm.032}$ \\
         & R$\uparrow$ & $.109_{\pm.058}$ & $.537_{\pm.048}$ & $.287_{\pm.101}$ & $.069_{\pm.121}$ & $.147_{\pm.029}$ \\
        \midrule
        \multirow[t]{4}{*}{Limit 50} & No retraining & $.341_{\pm.050}$ & $.400_{\pm.222}$ & $.287_{\pm.119}$ & $.543_{\pm.112}$ & $.150_{\pm.057}$ \\
         & PGD & $.112_{\pm.041}$ & $.561_{\pm.050}$ & $.229_{\pm.082}$ & $.021_{\pm.040}$ & $.142_{\pm.065}$ \\
         & R$\downarrow$ & $.144_{\pm.018}$ & $.596_{\pm.066}$ & $.212_{\pm.072}$ & $.012_{\pm.025}$ & $.141_{\pm.049}$ \\
         & R$\uparrow$ & $.126_{\pm.051}$ & $.443_{\pm.071}$ & $.225_{\pm.081}$ & $.052_{\pm.088}$ & $.137_{\pm.040}$ \\
        \midrule
        \multirow[t]{4}{*}{No overtaken} & No retraining & $.530_{\pm.044}$ & $.606_{\pm.201}$ & $.537_{\pm.123}$ & $.722_{\pm.064}$ & $.266_{\pm.122}$ \\
         & PGD & $.202_{\pm.067}$ & $.740_{\pm.039}$ & $.377_{\pm.129}$ & $.000_{\pm.000}$ & $.210_{\pm.087}$ \\
         & R$\downarrow$ & $.236_{\pm.044}$ & $.769_{\pm.030}$ & $.373_{\pm.130}$ & $.001_{\pm.003}$ & $.218_{\pm.091}$ \\
         & R$\uparrow$ & $.218_{\pm.047}$ & $.657_{\pm.054}$ & $.388_{\pm.118}$ & $.029_{\pm.048}$ & $.188_{\pm.082}$ \\
        \midrule
        \multirow[t]{4}{*}{Octagon} & No retraining & $.988_{\pm.005}$ & $.997_{\pm.001}$ & $.999_{\pm.000}$ & $1.0_{\pm.000}$ & $.999_{\pm.000}$ \\
         & PGD & $.795_{\pm.184}$ & $.981_{\pm.010}$ & $.883_{\pm.158}$ & $.687_{\pm.267}$ & $.933_{\pm.049}$ \\
         & R$\downarrow$ & $.861_{\pm.033}$ & $.984_{\pm.007}$ & $.841_{\pm.166}$ & $.635_{\pm.326}$ & $.944_{\pm.049}$ \\
         & R$\uparrow$ & $.743_{\pm.206}$ & $.973_{\pm.015}$ & $.912_{\pm.135}$ & $.705_{\pm.305}$ & $.952_{\pm.033}$ \\
        \midrule
        \multirow[t]{4}{*}{Red} & No retraining & $.999_{\pm.000}$ & $1.0_{\pm.000}$ & $1.0_{\pm.000}$ & $1.0_{\pm.000}$ & $1.0_{\pm.000}$ \\
         & PGD & $.971_{\pm.066}$ & $1.0_{\pm.000}$ & $.991_{\pm.017}$ & $.955_{\pm.028}$ & $.977_{\pm.030}$ \\
         & R$\downarrow$ & $.991_{\pm.003}$ & $1.0_{\pm.000}$ & $.959_{\pm.073}$ & $.973_{\pm.019}$ & $.977_{\pm.046}$ \\
         & R$\uparrow$ & $.906_{\pm.172}$ & $1.0_{\pm.000}$ & $.997_{\pm.004}$ & $.987_{\pm.010}$ & $.980_{\pm.043}$ \\
        \midrule
        \multirow[t]{4}{*}{Stop} & No retraining & $.596_{\pm.081}$ & $.653_{\pm.257}$ & $.457_{\pm.223}$ & $.904_{\pm.118}$ & $.196_{\pm.061}$ \\
         & PGD & $.126_{\pm.064}$ & $.760_{\pm.082}$ & $.356_{\pm.168}$ & $.011_{\pm.015}$ & $.176_{\pm.055}$ \\
         & R$\downarrow$ & $.157_{\pm.062}$ & $.770_{\pm.051}$ & $.322_{\pm.161}$ & $.018_{\pm.021}$ & $.191_{\pm.047}$ \\
         & R$\uparrow$ & $.108_{\pm.056}$ & $.677_{\pm.080}$ & $.366_{\pm.168}$ & $.077_{\pm.058}$ & $.163_{\pm.041}$ \\
        \midrule
        \multirow[t]{4}{*}{White with red border} & No retraining & $.998_{\pm.001}$ & $1.0_{\pm.000}$ & $1.0_{\pm.000}$ & $1.0_{\pm.000}$ & $1.0_{\pm.001}$ \\
         & PGD & $.958_{\pm.085}$ & $.999_{\pm.001}$ & $.980_{\pm.062}$ & $.888_{\pm.097}$ & $.986_{\pm.019}$ \\
         & R$\downarrow$ & $.990_{\pm.002}$ & $.999_{\pm.001}$ & $.982_{\pm.063}$ & $.843_{\pm.137}$ & $.989_{\pm.015}$ \\
         & R$\uparrow$ & $.928_{\pm.139}$ & $.999_{\pm.001}$ & $.981_{\pm.062}$ & $.970_{\pm.027}$ & $.992_{\pm.007}$ \\
        \midrule
        \multirow[t]{4}{*}{Yellow} & No retraining & $1.0_{\pm.000}$ & $1.0_{\pm.000}$ & $1.0_{\pm.000}$ & $1.0_{\pm.000}$ & $1.0_{\pm.000}$ \\
         & PGD & $.922_{\pm.205}$ & $1.0_{\pm.000}$ & $.984_{\pm.026}$ & $.930_{\pm.047}$ & $1.0_{\pm.000}$ \\
         & R$\downarrow$ & $.993_{\pm.004}$ & $1.0_{\pm.000}$ & $.951_{\pm.048}$ & $.934_{\pm.053}$ & $1.0_{\pm.000}$ \\
         & R$\uparrow$ & $.856_{\pm.239}$ & $1.0_{\pm.000}$ & $.990_{\pm.019}$ & $.984_{\pm.024}$ & $1.0_{\pm.000}$ \\
        \cline{1-7}
        \bottomrule
    \end{tabular}

    \label{tab:f1-gtsrb-result-original}
\end{table}

\begin{table}
    \caption{
        F1 scores \textbf{on original GTSRB concept data} of concept groups before and after adversarial retraining, reported as mean $\pm$ standard deviation across fifteen runs.
    }
    \centering
    \begin{tabular}{llccccc}
        \toprule
        Group & Retrained & LeNet-5 & M3 & ResNet-8 & VGG-16 & ViT \\
        \midrule
        \multirow[t]{4}{*}{Colour} & No retraining & $.999_{\pm.001}$ & $1.0_{\pm.000}$ & $1.0_{\pm.000}$ & $1.0_{\pm.000}$ & $1.0_{\pm.000}$ \\
         & PGD & $.950_{\pm.132}$ & $1.0_{\pm.001}$ & $.985_{\pm.039}$ & $.924_{\pm.069}$ & $.988_{\pm.022}$ \\
         & R$\downarrow$ & $.992_{\pm.003}$ & $1.0_{\pm.001}$ & $.964_{\pm.062}$ & $.917_{\pm.100}$ & $.989_{\pm.029}$ \\
         & R$\uparrow$ & $.896_{\pm.186}$ & $1.0_{\pm.001}$ & $.989_{\pm.037}$ & $.980_{\pm.022}$ & $.991_{\pm.026}$ \\
        \midrule
        \multirow[t]{4}{*}{Shape} & No retraining & $.988_{\pm.017}$ & $.999_{\pm.001}$ & $1.0_{\pm.000}$ & $1.0_{\pm.000}$ & $.999_{\pm.000}$ \\
         & PGD & $.857_{\pm.173}$ & $.993_{\pm.010}$ & $.922_{\pm.135}$ & $.782_{\pm.311}$ & $.951_{\pm.049}$ \\
         & R$\downarrow$ & $.909_{\pm.090}$ & $.994_{\pm.008}$ & $.909_{\pm.145}$ & $.773_{\pm.331}$ & $.956_{\pm.048}$ \\
         & R$\uparrow$ & $.822_{\pm.189}$ & $.989_{\pm.016}$ & $.934_{\pm.130}$ & $.795_{\pm.319}$ & $.971_{\pm.028}$ \\
        \midrule
        \multirow[t]{4}{*}{Symbol} & No retraining & $.456_{\pm.140}$ & $.530_{\pm.253}$ & $.448_{\pm.184}$ & $.653_{\pm.179}$ & $.223_{\pm.119}$ \\
         & PGD & $.146_{\pm.122}$ & $.685_{\pm.094}$ & $.329_{\pm.146}$ & $.009_{\pm.026}$ & $.193_{\pm.113}$ \\
         & R$\downarrow$ & $.181_{\pm.141}$ & $.709_{\pm.089}$ & $.305_{\pm.145}$ & $.015_{\pm.033}$ & $.204_{\pm.127}$ \\
         & R$\uparrow$ & $.150_{\pm.117}$ & $.605_{\pm.119}$ & $.333_{\pm.155}$ & $.044_{\pm.076}$ & $.180_{\pm.108}$ \\
        \cline{1-7}
        \midrule
        \multirow[t]{4}{*}{All} & No retraining & $.724_{\pm.287}$ & $.765_{\pm.295}$ & $.724_{\pm.305}$ & $.827_{\pm.215}$ & $.611_{\pm.398}$ \\
         & PGD & $.522_{\pm.403}$ & $.840_{\pm.170}$ & $.639_{\pm.336}$ & $.426_{\pm.454}$ & $.580_{\pm.397}$ \\
         & R$\downarrow$ & $.563_{\pm.399}$ & $.853_{\pm.157}$ & $.619_{\pm.341}$ & $.425_{\pm.453}$ & $.587_{\pm.395}$ \\
         & R$\uparrow$ & $.502_{\pm.387}$ & $.799_{\pm.212}$ & $.645_{\pm.340}$ & $.459_{\pm.457}$ & $.580_{\pm.408}$ \\
        \cline{1-7}
        \bottomrule
    \end{tabular}

    \label{tab:f1-gtsrb-result-summary}
\end{table}

%% file: tables/concept_progression/biased_gtsrb_f1.tex
\begin{table}
    \caption{
        F1 scores \textbf{on original GTSRB concept data} of each concept before and after adversarial retraining \textbf{for biased-trained models}, reported as mean $\pm$ standard deviation across fifteen runs.
    }
    \centering
    \begin{tabular}{llccccc}
        \toprule
         Concept & Retrained & LeNet-5 & M3 & ResNet-8 & VGG-16 & ViT \\
        \midrule
        \multirow[t]{4}{*}{Circle} & No retraining & $.973_{\pm.004}$ & $.998_{\pm.001}$ & $1.0_{\pm.000}$ & $1.0_{\pm.000}$ & $.999_{\pm.000}$ \\
         & PGD & $.388_{\pm.115}$ & $.977_{\pm.023}$ & $.814_{\pm.248}$ & $.781_{\pm.169}$ & $.784_{\pm.074}$ \\
         & R$\downarrow$ & $.412_{\pm.122}$ & $.987_{\pm.009}$ & $.822_{\pm.244}$ & $.746_{\pm.182}$ & $.822_{\pm.100}$ \\
         & R$\uparrow$ & $.386_{\pm.095}$ & $.985_{\pm.022}$ & $.816_{\pm.252}$ & $.779_{\pm.166}$ & $.790_{\pm.097}$ \\
        \midrule
        \multirow[t]{4}{*}{Diamond} & No retraining & $.997_{\pm.001}$ & $1.0_{\pm.000}$ & $1.0_{\pm.000}$ & $1.0_{\pm.000}$ & $1.0_{\pm.000}$ \\
         & PGD & $.594_{\pm.078}$ & $1.0_{\pm.000}$ & $.962_{\pm.037}$ & $.990_{\pm.008}$ & $.856_{\pm.106}$ \\
         & R$\downarrow$ & $.621_{\pm.134}$ & $.999_{\pm.001}$ & $.952_{\pm.078}$ & $.991_{\pm.007}$ & $.880_{\pm.081}$ \\
         & R$\uparrow$ & $.609_{\pm.112}$ & $1.0_{\pm.000}$ & $.979_{\pm.022}$ & $.993_{\pm.006}$ & $.862_{\pm.110}$ \\
        \midrule
        \multirow[t]{4}{*}{Inverted Triangle} & No retraining & $.998_{\pm.001}$ & $1.0_{\pm.000}$ & $1.0_{\pm.000}$ & $1.0_{\pm.000}$ & $.999_{\pm.001}$ \\
         & PGD & $.708_{\pm.102}$ & $1.0_{\pm.000}$ & $.987_{\pm.022}$ & $.993_{\pm.008}$ & $.957_{\pm.020}$ \\
         & R$\downarrow$ & $.745_{\pm.156}$ & $1.0_{\pm.000}$ & $.981_{\pm.056}$ & $.992_{\pm.015}$ & $.957_{\pm.039}$ \\
         & R$\uparrow$ & $.713_{\pm.085}$ & $1.0_{\pm.000}$ & $.987_{\pm.024}$ & $.995_{\pm.009}$ & $.955_{\pm.026}$ \\
        \midrule
        \multirow[t]{4}{*}{LKW no overtaken} & No retraining & $.275_{\pm.023}$ & $.618_{\pm.128}$ & $.639_{\pm.095}$ & $.907_{\pm.038}$ & $.693_{\pm.074}$ \\
         & PGD & $.080_{\pm.027}$ & $.428_{\pm.150}$ & $.335_{\pm.247}$ & $.042_{\pm.042}$ & $.254_{\pm.132}$ \\
         & R$\downarrow$ & $.131_{\pm.050}$ & $.425_{\pm.155}$ & $.322_{\pm.242}$ & $.031_{\pm.030}$ & $.285_{\pm.163}$ \\
         & R$\uparrow$ & $.095_{\pm.028}$ & $.461_{\pm.144}$ & $.334_{\pm.233}$ & $.035_{\pm.037}$ & $.252_{\pm.126}$ \\
        \midrule
        \multirow[t]{4}{*}{Limit 120} & No retraining & $.098_{\pm.016}$ & $.324_{\pm.111}$ & $.464_{\pm.096}$ & $.796_{\pm.011}$ & $.392_{\pm.046}$ \\
         & PGD & $.063_{\pm.014}$ & $.216_{\pm.115}$ & $.124_{\pm.064}$ & $.001_{\pm.005}$ & $.183_{\pm.057}$ \\
         & R$\downarrow$ & $.078_{\pm.021}$ & $.188_{\pm.087}$ & $.132_{\pm.071}$ & $.000_{\pm.000}$ & $.242_{\pm.071}$ \\
         & R$\uparrow$ & $.068_{\pm.015}$ & $.203_{\pm.101}$ & $.134_{\pm.074}$ & $.001_{\pm.004}$ & $.188_{\pm.057}$ \\
        \midrule
        \multirow[t]{4}{*}{Limit 30} & No retraining & $.130_{\pm.024}$ & $.336_{\pm.065}$ & $.295_{\pm.109}$ & $.843_{\pm.022}$ & $.293_{\pm.065}$ \\
         & PGD & $.092_{\pm.028}$ & $.235_{\pm.127}$ & $.168_{\pm.062}$ & $.009_{\pm.025}$ & $.141_{\pm.038}$ \\
         & R$\downarrow$ & $.097_{\pm.029}$ & $.259_{\pm.113}$ & $.161_{\pm.074}$ & $.006_{\pm.013}$ & $.175_{\pm.047}$ \\
         & R$\uparrow$ & $.090_{\pm.022}$ & $.224_{\pm.108}$ & $.173_{\pm.061}$ & $.007_{\pm.015}$ & $.139_{\pm.044}$ \\
        \midrule
        \multirow[t]{4}{*}{Limit 50} & No retraining & $.141_{\pm.025}$ & $.314_{\pm.067}$ & $.292_{\pm.094}$ & $.824_{\pm.088}$ & $.353_{\pm.045}$ \\
         & PGD & $.088_{\pm.033}$ & $.196_{\pm.058}$ & $.158_{\pm.063}$ & $.001_{\pm.003}$ & $.198_{\pm.081}$ \\
         & R$\downarrow$ & $.093_{\pm.023}$ & $.198_{\pm.063}$ & $.170_{\pm.061}$ & $.000_{\pm.000}$ & $.230_{\pm.098}$ \\
         & R$\uparrow$ & $.087_{\pm.029}$ & $.208_{\pm.056}$ & $.164_{\pm.061}$ & $.000_{\pm.000}$ & $.194_{\pm.088}$ \\
        \midrule
        \multirow[t]{4}{*}{No overtaken} & No retraining & $.205_{\pm.024}$ & $.516_{\pm.101}$ & $.526_{\pm.089}$ & $.889_{\pm.055}$ & $.541_{\pm.064}$ \\
         & PGD & $.085_{\pm.013}$ & $.423_{\pm.115}$ & $.199_{\pm.112}$ & $.003_{\pm.013}$ & $.160_{\pm.078}$ \\
         & R$\downarrow$ & $.137_{\pm.033}$ & $.384_{\pm.115}$ & $.213_{\pm.103}$ & $.001_{\pm.003}$ & $.193_{\pm.076}$ \\
         & R$\uparrow$ & $.108_{\pm.027}$ & $.408_{\pm.144}$ & $.226_{\pm.113}$ & $.015_{\pm.034}$ & $.162_{\pm.084}$ \\
        \midrule
        \multirow[t]{4}{*}{Octagon} & No retraining & $.974_{\pm.006}$ & $.996_{\pm.002}$ & $.999_{\pm.001}$ & $.999_{\pm.001}$ & $.998_{\pm.001}$ \\
         & PGD & $.380_{\pm.028}$ & $.969_{\pm.026}$ & $.743_{\pm.184}$ & $.725_{\pm.167}$ & $.803_{\pm.127}$ \\
         & R$\downarrow$ & $.419_{\pm.098}$ & $.969_{\pm.024}$ & $.705_{\pm.208}$ & $.645_{\pm.251}$ & $.884_{\pm.076}$ \\
         & R$\uparrow$ & $.394_{\pm.041}$ & $.974_{\pm.014}$ & $.789_{\pm.178}$ & $.672_{\pm.260}$ & $.833_{\pm.081}$ \\
        \midrule
        \multirow[t]{4}{*}{Red} & No retraining & $.998_{\pm.001}$ & $1.0_{\pm.000}$ & $1.0_{\pm.000}$ & $1.0_{\pm.000}$ & $1.0_{\pm.000}$ \\
         & PGD & $.564_{\pm.164}$ & $.999_{\pm.001}$ & $.984_{\pm.035}$ & $.989_{\pm.006}$ & $.848_{\pm.052}$ \\
         & R$\downarrow$ & $.640_{\pm.143}$ & $.997_{\pm.007}$ & $.943_{\pm.165}$ & $.958_{\pm.035}$ & $.864_{\pm.075}$ \\
         & R$\uparrow$ & $.600_{\pm.193}$ & $.999_{\pm.001}$ & $.995_{\pm.009}$ & $.995_{\pm.005}$ & $.850_{\pm.085}$ \\
        \midrule
        \multirow[t]{4}{*}{Stop} & No retraining & $.164_{\pm.046}$ & $.411_{\pm.111}$ & $.501_{\pm.220}$ & $.977_{\pm.016}$ & $.500_{\pm.129}$ \\
         & PGD & $.077_{\pm.016}$ & $.323_{\pm.130}$ & $.281_{\pm.117}$ & $.048_{\pm.090}$ & $.161_{\pm.051}$ \\
         & R$\downarrow$ & $.085_{\pm.020}$ & $.278_{\pm.115}$ & $.301_{\pm.128}$ & $.021_{\pm.035}$ & $.218_{\pm.068}$ \\
         & R$\uparrow$ & $.077_{\pm.018}$ & $.286_{\pm.149}$ & $.322_{\pm.108}$ & $.064_{\pm.100}$ & $.170_{\pm.041}$ \\
        \midrule
        \multirow[t]{4}{*}{White with red border} & No retraining & $.995_{\pm.001}$ & $1.0_{\pm.000}$ & $1.0_{\pm.000}$ & $1.0_{\pm.000}$ & $.999_{\pm.001}$ \\
         & PGD & $.611_{\pm.174}$ & $.944_{\pm.095}$ & $.925_{\pm.100}$ & $.871_{\pm.144}$ & $.945_{\pm.031}$ \\
         & R$\downarrow$ & $.696_{\pm.208}$ & $.979_{\pm.026}$ & $.912_{\pm.101}$ & $.754_{\pm.245}$ & $.896_{\pm.115}$ \\
         & R$\uparrow$ & $.618_{\pm.174}$ & $.971_{\pm.054}$ & $.946_{\pm.087}$ & $.891_{\pm.076}$ & $.930_{\pm.040}$ \\
        \midrule
        \multirow[t]{4}{*}{Yellow} & No retraining & $1.0_{\pm.001}$ & $1.0_{\pm.000}$ & $1.0_{\pm.000}$ & $1.0_{\pm.000}$ & $1.0_{\pm.000}$ \\
         & PGD & $.506_{\pm.302}$ & $1.0_{\pm.000}$ & $.999_{\pm.002}$ & $.855_{\pm.068}$ & $.998_{\pm.001}$ \\
         & R$\downarrow$ & $.592_{\pm.206}$ & $1.0_{\pm.000}$ & $.990_{\pm.012}$ & $.789_{\pm.147}$ & $.998_{\pm.001}$ \\
         & R$\uparrow$ & $.576_{\pm.138}$ & $1.0_{\pm.000}$ & $.998_{\pm.003}$ & $.901_{\pm.059}$ & $.998_{\pm.001}$ \\
        \cline{1-7}
        \bottomrule
    \end{tabular}

    \label{tab:f1-biased-gtsrb-result-original}
\end{table}

\begin{table}
    \caption{
        F1 scores \textbf{on original GTSRB concept data} of concept groups before and after adversarial retraining \textbf{for biased-trained models}, reported as mean $\pm$ standard deviation across fifteen runs.
    }
    \centering
    \begin{tabular}{llccccc}
        \toprule
         Group & Retrained & LeNet-5 & M3 & ResNet-8 & VGG-16 & ViT \\
        \midrule
        \multirow[t]{4}{*}{Colour} & No retraining & $.997_{\pm.002}$ & $1.0_{\pm.000}$ & $1.0_{\pm.000}$ & $1.0_{\pm.000}$ & $1.0_{\pm.000}$ \\
         & PGD & $.560_{\pm.222}$ & $.981_{\pm.060}$ & $.969_{\pm.068}$ & $.905_{\pm.109}$ & $.930_{\pm.071}$ \\
         & R$\downarrow$ & $.643_{\pm.189}$ & $.992_{\pm.018}$ & $.948_{\pm.114}$ & $.834_{\pm.186}$ & $.919_{\pm.097}$ \\
         & R$\uparrow$ & $.598_{\pm.167}$ & $.990_{\pm.033}$ & $.980_{\pm.055}$ & $.929_{\pm.072}$ & $.926_{\pm.081}$ \\
        \midrule
        \multirow[t]{4}{*}{Shape} & No retraining & $.986_{\pm.013}$ & $.998_{\pm.002}$ & $1.0_{\pm.000}$ & $1.0_{\pm.001}$ & $.999_{\pm.001}$ \\
         & PGD & $.517_{\pm.164}$ & $.986_{\pm.022}$ & $.877_{\pm.183}$ & $.872_{\pm.168}$ & $.850_{\pm.112}$ \\
         & R$\downarrow$ & $.549_{\pm.190}$ & $.989_{\pm.018}$ & $.865_{\pm.197}$ & $.844_{\pm.216}$ & $.886_{\pm.089}$ \\
         & R$\uparrow$ & $.525_{\pm.165}$ & $.990_{\pm.017}$ & $.893_{\pm.176}$ & $.860_{\pm.206}$ & $.860_{\pm.103}$ \\
        \midrule
        \multirow[t]{4}{*}{Symbol} & No retraining & $.170_{\pm.064}$ & $.421_{\pm.143}$ & $.459_{\pm.167}$ & $.857_{\pm.082}$ & $.444_{\pm.151}$ \\
         & PGD & $.080_{\pm.025}$ & $.273_{\pm.164}$ & $.211_{\pm.140}$ & $.017_{\pm.043}$ & $.167_{\pm.088}$ \\
         & R$\downarrow$ & $.100_{\pm.036}$ & $.269_{\pm.148}$ & $.216_{\pm.140}$ & $.011_{\pm.025}$ & $.205_{\pm.103}$ \\
         & R$\uparrow$ & $.086_{\pm.025}$ & $.274_{\pm.162}$ & $.224_{\pm.139}$ & $.020_{\pm.048}$ & $.170_{\pm.087}$ \\
        \cline{1-7}
        \midrule
        \multirow[t]{4}{*}{All} & No retraining & $.581_{\pm.414}$ & $.710_{\pm.307}$ & $.730_{\pm.295}$ & $.928_{\pm.092}$ & $.722_{\pm.298}$ \\
         & PGD & $.308_{\pm.266}$ & $.629_{\pm.376}$ & $.564_{\pm.382}$ & $.452_{\pm.448}$ & $.526_{\pm.372}$ \\
         & R$\downarrow$ & $.345_{\pm.282}$ & $.630_{\pm.376}$ & $.558_{\pm.377}$ & $.425_{\pm.439}$ & $.552_{\pm.362}$ \\
         & R$\uparrow$ & $.321_{\pm.265}$ & $.632_{\pm.377}$ & $.577_{\pm.381}$ & $.455_{\pm.452}$ & $.529_{\pm.372}$ \\
        \cline{1-7}
        \bottomrule
    \end{tabular}

    \label{tab:f1-biased-gtsrb-result-summary}
\end{table}

%% file: tables/rule_table-original.tex
\begin{table*}[htb]
    \centering
    \caption{Rule compliance on original dataset}
    \begin{tabular}{llcccc}
        \toprule
        \multicolumn{2}{c}{Backbone} &  Original & PGD & R$\downarrow$ & \textcolor{gray}{R$\uparrow$} \\
        \midrule
        \multirow{5}{*}{\rotatebox[origin=c]{90}{MNIST}}
        & LENET-5 & $.926_{\pm.006}$ & $.890_{\pm.022}$ & $.895_{\pm.029}$ & \textcolor{gray}{$.887_{\pm.023}$} \\
        & M3 & $.982_{\pm.002}$ & $.957_{\pm.014}$ & $.957_{\pm.014}$ & \textcolor{gray}{$.964_{\pm.009}$} \\
        & ResNet-8 & $.987_{\pm.002}$ & $.973_{\pm.005}$ & $.975_{\pm.006}$ & \textcolor{gray}{ $.978_{\pm.004}$} \\
        & ViT & $.958_{\pm.006}$ & $.903_{\pm.017}$ & $.907_{\pm.014}$ & \textcolor{gray}{$.903_{\pm.014}$} \\
        & VGG-16 & $.985_{\pm.001}$ & $.897_{\pm.022}$ & $.899_{\pm.010}$ & \textcolor{gray}{$.889_{\pm.053}$} \\
        \midrule
        \multirow{5}{*}{\rotatebox[origin=c]{90}{\quad GTSRB}}
        & LENET-5 & $.937_{\pm.021}$ & $.900_{\pm.029}$ & $.916_{\pm.028}$ & \textcolor{gray}{$.906_{\pm.024}$} \\
        & M3 & $.938_{\pm.007}$ & $.933_{\pm.016}$ & $.929_{\pm.019}$ & \textcolor{gray}{$.928_{\pm.014}$} \\
        & ResNet-8 & $.941_{\pm.012}$ & $.940_{\pm.013}$ & $.939_{\pm.012}$ & \textcolor{gray}{$.940_{\pm.012}$} \\
        & ViT & $.827_{\pm.010}$ & $.825_{\pm.036}$ & $.828_{\pm.035}$ & \textcolor{gray}{$.805_{\pm.028}$} \\
        & VGG-16 & $.910_{\pm.008}$ & $.955_{\pm.001}$ & $.956_{\pm.000}$ & \textcolor{gray}{$.955_{\pm.001}$} \\
        \midrule
        \multirow{5}{*}{\rotatebox[origin=c]{90}{\quad Biased}}
        & LENET-5 & $.925_{\pm.024}$ & $.858_{\pm.036}$ & $.856_{\pm.078}$ & \textcolor{gray}{$.853_{\pm.031}$} \\
        & M3 & $.945_{\pm.004}$ & $.931_{\pm.019}$ & $.934_{\pm.014}$ & \textcolor{gray}{$.935_{\pm.014}$} \\
        & ResNet-8 & $.926_{\pm.019}$ & $.938_{\pm.018}$ & $.936_{\pm.018}$ & \textcolor{gray}{$.938_{\pm.016}$} \\
        & ViT & $.869_{\pm.067}$ & $.830_{\pm.020}$ & $.830_{\pm.020}$ & \textcolor{gray}{$.819_{\pm.040}$} \\
        & VGG-16 & $.908_{\pm.018}$ & $.956_{\pm.000}$ & $.956_{\pm.000}$ & \textcolor{gray}{$.956_{\pm.000}$} \\
        \bottomrule
    \end{tabular}
    \label{tab:rule-compliance-result}
\end{table*}

%% file: tables/rule_table-RCT-AE.tex
\begin{table*}[tbh]
    \centering
    \caption{Rule compliance on adversarial examples of R$\downarrow$}
    \begin{tabular}{llccc}
        \toprule
        \multicolumn{2}{c}{Backbone} & Original & PGD & R$\downarrow$ \\
        \midrule
        \multirow{5}{*}{\rotatebox[origin=c]{90}{MNIST}}
        & LENET-5 & $.830_{\pm.037}$ & $.872_{\pm.020}$ & $.851_{\pm.021}$ \\
        & M3 & $.855_{\pm.011}$ & $.962_{\pm.012}$ & $.954_{\pm.013}$ \\
        & ResNet-8 & $.939_{\pm.014}$ & $.972_{\pm.006}$ & $.974_{\pm.006}$ \\
        & ViT & $.903_{\pm.013}$ & $.904_{\pm.019}$ & $.904_{\pm.011}$ \\
        & VGG-16 & $.821_{\pm.005}$ & $.897_{\pm.022}$ & $.899_{\pm.010}$ \\
        \midrule
        \multirow{5}{*}{\rotatebox[origin=c]{90}{\quad GTSRB}}
        & LENET-5 & $.212_{\pm.035}$ & $.839_{\pm.053}$ & $.841_{\pm.045}$ \\
        & M3 & $.850_{\pm.071}$ & $.888_{\pm.048}$ & $.890_{\pm.046}$ \\
        & ResNet-8 & $.606_{\pm.116}$ & $.955_{\pm.011}$ & $.952_{\pm.011}$ \\
        & ViT & $.272_{\pm.072}$ & $.653_{\pm.079}$ & $.563_{\pm.056}$ \\
        & VGG-16 & $.879_{\pm.001}$ & $.956_{\pm.001}$ & $.956_{\pm.001}$ \\
        \midrule
        \multirow{5}{*}{\rotatebox[origin=c]{90}{\quad Biased}}
        & LENET-5 & $.184_{\pm.073}$ & $.805_{\pm.046}$ & $.768_{\pm.126}$ \\
        & M3 & $.777_{\pm.035}$ & $.841_{\pm.028}$ & $.842_{\pm.029}$ \\
        & ResNet-8 & $.616_{\pm.064}$ & $.954_{\pm.016}$ & $.951_{\pm.014}$ \\
        & ViT & $.441_{\pm.098}$ & $.675_{\pm.098}$ & $.645_{\pm.068}$ \\
        & VGG-16 & $.879_{\pm.006}$ & $.956_{\pm.001}$ & $.955_{\pm.001}$ \\
        \bottomrule
    \end{tabular}
    \label{tab:rule-compliance-result-rct-ae}
\end{table*}

%% file: tables/rule_table-IRCT-AE.tex
\begin{table*}[h]
    \centering
    \caption{Rule compliance on adversarial examples of R$\uparrow$}
    \begin{tabular}{llccc}
        \toprule
        \multicolumn{2}{c}{Backbone} & Original & PGD & R$\uparrow$\\
        \midrule
        \multirow{5}{*}{\rotatebox[origin=c]{90}{MNIST}}
        & LENET-5 & $.972_{\pm.008}$ & $.898_{\pm.025}$ & $.922_{\pm.029}$ \\
        & M3 & $.954_{\pm.016}$ & $.952_{\pm.014}$ & $.965_{\pm.007}$ \\
        & ResNet-8 & $.946_{\pm.009}$ & $.970_{\pm.008}$ & $.974_{\pm.005}$ \\
        & ViT & $.923_{\pm.006}$ & $.905_{\pm.019}$ & $.903_{\pm.012}$ \\
        & VGG-16 & $.967_{\pm.002}$ & $.897_{\pm.022}$ & $.890_{\pm.053}$ \\
        \midrule
        \multirow{5}{*}{\rotatebox[origin=c]{90}{ GTSRB}}
        & LENET-5 & $.982_{\pm.005}$ & $.942_{\pm.015}$ & $.943_{\pm.019}$ \\
        & M3 & $.993_{\pm.001}$ & $.953_{\pm.005}$ & $.952_{\pm.005}$ \\
        & ResNet-8 & $.998_{\pm.001}$ & $.964_{\pm.010}$ & $.962_{\pm.006}$ \\
        & ViT & $.997_{\pm.002}$ & $.885_{\pm.038}$ & $.875_{\pm.042}$ \\
        & VGG-16 & $1.0_{\pm.000}$ & $.957_{\pm.001}$ & $.956_{\pm.001}$ \\
        \midrule
        \multirow{5}{*}{\rotatebox[origin=c]{90}{ Biased}}
        & LENET-5 & $.990_{\pm.004}$ & $.889_{\pm.043}$ & $.891_{\pm.031}$ \\
        & M3 & $.992_{\pm.005}$ & $.953_{\pm.009}$ & $.953_{\pm.007}$ \\
        & ResNet-8 & $.947_{\pm.018}$ & $.951_{\pm.017}$ & $.951_{\pm.016}$ \\
        & ViT & $.976_{\pm.018}$ & $.885_{\pm.026}$ & $.884_{\pm.035}$ \\
        & VGG-16 & $1.0_{\pm.000}$ & $.958_{\pm.001}$ & $.957_{\pm.001}$ \\
        \bottomrule
    \end{tabular}
    \label{tab:rule-compliance-result-irct-ae}
\end{table*}